\documentclass[a4paper,UKenglish,cleveref, autoref, thm-restate]{lipics-v2021}
\nolinenumbers

	\pdfoutput=1 %
	\hideLIPIcs  %

\usepackage[utf8]{inputenc}
\usepackage{csquotes}
\usepackage{lmodern} %
\usepackage[T1]{fontenc} %

\usepackage{multirow}

\usepackage[ruled,vlined,linesnumbered]{algorithm2e}
	\SetKwInOut{Input}{input}
	\SetKwInOut{Output}{output}

\usepackage{subcaption}

\usepackage{amsmath} %
\usepackage{amssymb} %

\usepackage[misc,geometry]{ifsym} %

\usepackage{thm-restate}

\usepackage{paralist} %

\newenvironment{ienumerate}
	{\ifdefined\VersionLong\begin{enumerate}\else\begin{inparaenum}[\itshape i\upshape)]\fi}
	{\ifdefined\VersionLong\end{enumerate}\else\end{inparaenum}\fi}

\newenvironment{oneenumerate}
	{\ifdefined\VersionLong\begin{enumerate}\else\begin{inparaenum}[1)]\fi}
	{\ifdefined\VersionLong\end{enumerate}\else\end{inparaenum}\fi}

\usepackage[svgnames,table]{xcolor}
\definecolor{LSdeuxNrichblack}{HTML}{020A13}
\definecolor{LSdeuxNbazul}{HTML}{1273CE}
\definecolor{LSdeuxNcornflower}{HTML}{5297FF}
\definecolor{LSdeuxNAntiflashwhite}{HTML}{F8FBFF}

\newcommand{\fakeParagraph}[1]{\smallskip

\noindent\textbf{\textsf{#1.}}}

\newcommand{\defProblem}[3]
{%
\noindent\fcolorbox{LSdeuxNrichblack}{LSdeuxNAntiflashwhite}{
	\begin{minipage}{.95\columnwidth}
		\textbf{#1:}\\ %
		\textsc{Input}: #2\\
		\textsc{Problem}: #3
	\end{minipage}
}

	\smallskip

}

\usepackage{tikz}
\usetikzlibrary{arrows,automata,positioning}
\tikzstyle{every node}     = [initial text=]
\tikzstyle{PTA}            = [auto, ->, >=stealth']
\tikzstyle{location}       = [rectangle, rounded corners, minimum size= 12pt, draw=black, fill=blue!10, inner sep=2pt]
\tikzstyle{location2CM}    = [location,thick,fill=green!30]
\tikzstyle{error}          = [fill=red!50]
\tikzstyle{sink}           = [fill=gray!50]
\tikzstyle{invariant}      = [draw=black, dotted, fill=yellow, inner sep= 1pt, node distance=0] %
\tikzstyle{invariantNorth} = [invariant, yshift=1em]
\tikzstyle{invariantSouth} = [invariant, yshift=-1em]
\tikzstyle{invariantWest}  = [invariant, xshift=-1.5em]
\tikzstyle{edge}           = [->,  >=stealth']
\tikzstyle{locguard}       = [fill=cyan!30, inner sep=1pt, node distance=0] %

\definecolor{coloract}  {rgb}{0  , 0.4, 0}
\definecolor{colorclock}{rgb}{0.7, 0  , 0}
\definecolor{colordisc} {rgb}{1  , 0  , 1}
\definecolor{colorloc}  {rgb}{0.1, 0.1, 0.1}
\definecolor{colorparam}{rgb}{1, 0.6  , 0.0}

\newcommand{\cellHeader}[1]{\cellcolor{LSdeuxNAntiflashwhite}\textbf{#1}}
\newcommand{\rowHeader}{\rowcolor{LSdeuxNAntiflashwhite}}

\newcommand{\cellOpen}{\cellcolor{yellow!20}\textbf{$?$}}

\newcommand{\cellDecidable}{\cellcolor{green!35}$\surd$}
\newcommand{\cellUndecidable}{\cellcolor{red!35}$\times$}

\newcommand{\set}[1]{\ensuremath{\left\{#1\right\}}}

\newif\ifinfigure
\let\infigure\iffalse %

\let\oldtikzpicture\tikzpicture
\let\endoldtikzpicture\endtikzpicture
\renewenvironment{tikzpicture}
  {\let\ifinfigure\iftrue\oldtikzpicture}
  {\endoldtikzpicture\let\ifinfigure\iffalse}

\newcommand{\styleAutomaton}[1]{\ensuremath{\mathcal{#1}}}

\newcommand{\styleact}[1]{\ensuremath{
	\ifinfigure%
		\textcolor{coloract}{{#1}}%
	\else%
		#1%
	\fi%
	}%
}

\newcommand{\styleclock}[1]{\ensuremath{
	\ifinfigure%
		\textcolor{colorclock}{{#1}}%
	\else%
		#1%
	\fi%
	}%
}

\newcommand{\styleloc}[1]{\ensuremath{
	\ifinfigure%
		\textcolor{colorloc}{{\mathrm{#1}}}%
	\else%
		\mathrm{#1}%
	\fi%
	}%
}

\newcommand{\styleparam}[1]{\ensuremath{
	\ifinfigure%
		\textcolor{colorparam}{{#1}}%
	\else%
		#1%
	\fi%
	}%
}

\newcommand{\init}{\ensuremath{0}}
\newcommand{\final}{\ensuremath{f}}

\newcommand{\assign}{\leftarrow}
\newcommand{\compOp}{\bowtie}

\newcommand{\decrincr}{\preceq}

\newcommand{\reset}[2]{\ensuremath{[#1]_{#2}}}
\newcommand{\valuate}[2]{\ensuremath{#2(#1)}}

\newcommand{\setN}{\ensuremath{\mathbb{N}}}
\newcommand{\setNpos}{\ensuremath{\setN_{> 0}}}
\newcommand{\setQ}{\ensuremath{\mathbb{Q}}}
\newcommand{\setQpos}{\ensuremath{\setQ_{> 0}}}
\newcommand{\setR}{\ensuremath{\mathbb{R}}}

\newcommand{\setRgeqzero}{\ensuremath{\setR_{\geq 0}}}

\newcommand{\setZ}{\ensuremath{\mathbb{Z}}}

\newcommand{\clock}{\ensuremath{\styleclock{x}}}
\newcommand{\clockx}{\ensuremath{\styleclock{x}}}

\newcommand{\clocki}[1]{\ensuremath{\styleclock{\clock_{#1}}}}
\newcommand{\ClockCard}{H} %
\newcommand{\clockval}{\ensuremath{\mu}}
\newcommand{\ClockSet}{\ensuremath{\mathbb{X}}} %
\newcommand{\ClocksZero}{\ensuremath{\vec{0}}}

\newcommand{\resets}{\ensuremath{R}}

\newcommand{\param}{\ensuremath{\styleparam{p}}}
\newcommand{\paramp}{\ensuremath{\styleparam{\param}}}
\newcommand{\parami}[1]{\ensuremath{\styleparam{\param_{#1}}}}
\newcommand{\ParamCard}{\ensuremath{M}} %
\newcommand{\pval}{\ensuremath{v}}
\newcommand{\ParamSet}{\ensuremath{\mathbb{P}}} %
\newcommand{\ParamSetL}{\ensuremath{\ParamSet_L}}
\newcommand{\ParamSetU}{\ensuremath{\ParamSet_U}}
\newcommand{\pvalzerod}[1]{\ensuremath{\pval_{0/#1}}}
\newcommand{\pvalzeroinf}{\ensuremath{\pvalzerod{\infty}}}

\newcommand{\CTrue}{\ensuremath{\mathit{True}}}

\newcommand{\ConstraintsXP}{\ensuremath{\Phi(\ClockSet \cup \ParamSet)}}

\newcommand{\calM}{\ensuremath{\mathcal{M}}}

\newcommand{\countervalue}{\ensuremath{c}}
\newcommand{\Counter}{\ensuremath{\mathcal{C}}}
\newcommand{\maxcounterval}{\ensuremath{c_\mathit{max}}}
\newcommand{\counterGenericIndex}{\ensuremath{l}}

\newcommand{\cmshaltname}{\ensuremath{\textrm{halt}}} %
\newcommand{\cms}{\ensuremath{\mathtt{q}}} %
\newcommand{\cmshalt}{\ensuremath{\cms_{\cmshaltname}}} %
\newcommand{\cmspta}{\ensuremath{q}} %

\newcommand{\clockCMt}{\ensuremath{\styleclock{t}}}
\newcommand{\clockCMcone}{\ensuremath{\styleclock{\clock_1}}}
\newcommand{\clockCMctwo}{\ensuremath{\styleclock{\clock_2}}}

\newcommand{\TA}{\ensuremath{{A}}}
\newcommand{\PTA}{\ensuremath{\styleAutomaton{A}}}

\newcommand{\action}{\ensuremath{\styleact{a}}}
\newcommand{\ActionSet}{\ensuremath{\Sigma}}

\newcommand{\ADRaction}{\ensuremath{\styleact{\sigma}}}
\newcommand{\ADRactioni}[1]{\ensuremath{\styleact{\ADRaction_{#1}}}}

\newcommand{\constraint}{\ensuremath{C}}

\newcommand{\EdgeSet}{\ensuremath{E}}

\newcommand{\guard}{\ensuremath{g}}

\newcommand{\invariant}{\ensuremath{I}}

\newcommand{\loc}{\ensuremath{\styleloc{\ell}}}
\newcommand{\lochalt}{\ensuremath{\cmspta_{\cmshaltname}}}
\newcommand{\loci}[1]{\ensuremath{\styleloc{\loc}_{#1}}}
\newcommand{\locij}[2]{\ensuremath{\styleloc{\loc}_{#1}^{#2}}}
\newcommand{\locinit}{\ensuremath{\styleloc{\loc_\init}}}
\newcommand{\locfinal}{\ensuremath{\styleloc{\loc_\final}}}
\newcommand{\LocSet}{\ensuremath{L}} %

\newcommand{\locguardof}[1]{[#1]}
\newcommand{\locguard}{\ensuremath{\gamma}}

\newcommand{\locinitexample}{\ensuremath{\mathsf{init}}}
\newcommand{\loclisten}{\ensuremath{\mathsf{listen}}}
\newcommand{\locpost}{\ensuremath{\mathsf{post}}}
\newcommand{\locreading}{\ensuremath{\mathsf{reading}}}
\newcommand{\locdone}{\ensuremath{\mathsf{done}}}
\newcommand{\locerror}{\ensuremath{\mathsf{error}}}
\newcommand{\longuefleche}[1]{\stackrel{#1}{\longrightarrow}}

\newcommand{\cDelay}{\ensuremath{d}}

\newcommand{\semantics}[1]{\ensuremath{\mathfrak{T}_{#1}}}

\newcommand{\transition}{{\ensuremath{\rightarrow}}}
\newcommand{\transitionWith}[1]{\stackrel{#1}{\mapsto}}

\newcommand{\StateSet}{\ensuremath{\mathfrak{S}}}
\newcommand{\concstate}{\ensuremath{\mathfrak{s}}}

\newcommand{\concstateinit}{\ensuremath{\concstate_\init}}

\newcommand{\gTA}{gTA}
\newcommand{\gPTA}{gPTA}
\newcommand{\gPTAs}{gPTAs}

\newcommand{\DTN}{DTN}
\newcommand{\DTNs}{DTNs}
\newcommand{\PDTN}{PDTN}
\newcommand{\PDTNs}{PDTNs}

\newcommand{\NTAs}{NTAs}
\newcommand{\NPTA}{NPTA}
\newcommand{\NPTAs}{NPTAs}

\newcommand{\gEdgeSet}{\ensuremath{E}}
\newcommand{\gtatran}{\ensuremath{\tau}}

\newcommand{\network}[2]{\ensuremath{{#1}^{#2}}}
\newcommand{\Pnetwork}[2]{\ensuremath{{#1}^{#2}}}
\newcommand{\nConfig}{\ensuremath{\mathfrak{c}}} %
\newcommand{\ConfigsSet}{\ensuremath{\mathfrak{C}}}
\newcommand{\nConfigInit}{\ensuremath{\hat{\nConfig}}}
\newcommand{\dtnconfig}{\ensuremath{\big((\loc_1,\clockval_1), \ldots, (\loc_n,\clockval_n)\big)}}
\newcommand{\dtnconfigdelay}{\ensuremath{\big((\loc_1,\clockval_1+\cDelay),\ldots,(\loc_n,\clockval_n+\cDelay)\big)}}
\newcommand{\dtnconfigsucc}{\ensuremath{\big((\loc_1',\clockval_1'),\ldots,(\loc_n',\clockval_n')\big)}}

\newcommand{\tatran}{\ensuremath{\tau}}

\newcommand{\transitionOf}[3]{{#1}\xrightarrow{#2}{#3}}

\newcommand{\gcomputation}{\ensuremath{\pi}}

\newcommand{\generalFigsScaleFactor}{1}

\newcommand{\ComplexityFont}[1]{{\sffamily\upshape #1}}

\newcommand{\NEXPTIME}{\ComplexityFont{NEXPTIME}}
\newcommand{\threeNEXPTIME}{\ComplexityFont{3-NEXPTIME}}

\newcommand{\EXPSPACE}{\ComplexityFont{EXPSPACE}}
\newcommand{\twoEXPSPACE}{\ComplexityFont{2-EXPSPACE}}

\newcommand{\eg}{e.g.,\xspace}

\newcommand{\ie}{i.e.,\xspace}
\newcommand{\suchthat}{s.t.\xspace}
\newcommand{\viz}{viz.,\xspace}
\newcommand{\WLOG}{w.l.o.g.\xspace}
\newcommand{\wrt}{w.r.t.\xspace}

\title{Parametric disjunctive timed networks}

\author{\'Etienne Andr\'e}{Université Sorbonne Paris Nord, LIPN, CNRS UMR 7030, Villetaneuse, France \and Institut universitaire de France (IUF), France \and Nantes Université, CNRS, LS2N, Nantes, France \and \url{https://lipn.univ-paris13.fr/~andre/}}{}{https://orcid.org/0000-0001-8473-9555}{Partially supported by ANR TAPAS (ANR-24-CE25-5742).}%
\author{Swen Jacobs}{CISPA Helmholtz Center for Information Security, Germany \and \url{https://swenjacobs.github.io}}{}{https://orcid.org/0000-0002-9051-4050}{}
\author{Engel Lefaucheux}{Université de Lorraine, CNRS, Inria, LORIA, F-54000 Nancy, France \and \url{https://elefauch.github.io}}{}{https://orcid.org/0000-0003-0875-300X}{Partially supported by ANR GUMMIS (ANR-25-CE48-5096).}
\Copyright{\'Etienne Andr\'e, Swen Jacobs and Engel Lefaucheux} %
\authorrunning{{\'E}.\ André, S.\ Jacobs and E.\ Lefaucheux} %



	\relatedversion{This is the author version of the manuscript of the same name published in the proceedings of the 34th EACSL Annual Conference on Computer Science Logic (CSL~2026).
	Compared to CSL~2026, this version includes the full proof of \cref{theorem:decidability:no-invariant}.
	} %

\funding{This work is partially supported by ANR BisoUS (ANR-22-CE48-0012).}
\EventEditors{John Q. Open and Joan R. Access}
\EventNoEds{2}
\EventLongTitle{42nd Conference on Very Important Topics (CVIT 2016)}
\EventShortTitle{CVIT 2016}
\EventAcronym{CVIT}
\EventYear{2016}
\EventDate{December 24--27, 2016}
\EventLocation{Little Whinging, United Kingdom}
\EventLogo{}
\SeriesVolume{42}
\ArticleNo{1}

\begin{document}
\sloppy

\maketitle{}
\begin{abstract}
	We consider distributed systems with an arbitrary number of processes, modelled by timed automata that communicate through location guards: a process can take a guarded transition if at least one other process is in a given location.
	In this work, we introduce parametric disjunctive timed networks, where each timed automaton may contain timing parameters, \ie{} unknown constants.
	We investigate two problems: deciding the emptiness of the set of parameter valuations for which
		1) a given location is reachable for at least one process (local property), and
		2) a global state is reachable where all processes are in a given location (global property).
	Our main positive result is that the first problem is decidable for networks of processes with a single clock and without invariants; this result holds for arbitrarily many timing parameters---a setting with few known decidability results.
	However, it becomes undecidable when invariants are allowed, or when considering global properties, even for systems with a single parameter.
	This highlights the significant expressive power of invariants in these networks.
	Additionally, we exhibit further decidable subclasses by restraining the syntax of guards and invariants.

\end{abstract}

\newpage

\section{Introduction}\label{section:introduction}
Parametrised verification~\cite{BJKKRVW15,AD16,AST18} consists in verifying a system's behaviour across all possible configurations of a certain parameter, such as the number of processes.
It is most commonly used in the context of networks of identical finite-state processes, and it involves proving that a property holds for any number of processes.
This type of verification is crucial for distributed systems, where an arbitrary number of identical agents may be interacting, and ensures that system correctness is maintained no matter the scale of the system.

When timing constraints are involved, more powerful formalisms are needed.
Timed automata (TAs)~\cite{AD94} extend finite-state automata with clocks (measuring the time elapsing), and offer a powerful framework for the specification and verification of real-time systems.
Clock constraints are used to constrain the time to remain in a location (``invariant'') or to take a transition (``guard'').

Several works consider parametrised verification for networks of timed automata~\cite{AJ03,ADM04,AAC12,ADRST16}, showing that it quickly hits undecidability, notably when multiple clocks are involved.
Decidability in the presence of multiple clocks can be preserved by restricting the communication between processes, \eg{} to communication via \emph{location guards}: a process can take a transition guarded by a location $\loc$ if at least one other process currently occupies $\loc$.
In disjunctive timed networks, where identical processes communicate via such location guards, local reachability and safety properties can be decided for any number of clocks~\cite{SS20,AEJK24}, even in the presence of invariants~\cite{AJKS25}.

When timing constants are not known with full precision (or completely unknown, \eg{} at the beginning of the design phase), timed automata may become impractical.
Parametric timed automata (PTAs)~\cite{AHV93} address this issue by allowing the modelling and verification of real-time systems with unknown or variable timing constants modelled as \emph{timing parameters}.
This flexibility enables the analysis of system behaviour across a range of parameter valuations, ensuring correctness under diverse conditions and facilitating optimization of parameters.

Common decision problems for parametric timed automata also quickly hit undecidability: emptiness of the parameter valuations set for which a given location is reachable (``reachability-emptiness''), for a single PTA, is undecidable with as few as 3~clocks and a single timing parameter (see \cite{Andre19STTT} for a survey).

\subsection*{Contributions}
In this paper, we address the verification of systems with unknown timing constants over an arbitrary number of processes.
In that sense, this parametrised parametric timed setting can be seen as having parameters in two dimensions: timing parameters, and number of processes.
We introduce \emph{parametric disjunctive timed networks} (\PDTNs{}) as networks of identical parametric timed processes resembling PTAs, and communicating via location guards.
A combination of two types of parameters appears natural, especially when designing and verifying communication protocols.
These protocols must function regardless of the number of participants (hence the parametric size of networks), while timing parameters allow designers to adjust critical time constraints in each process during early stages of development.

\fakeParagraph{Motivating example}
We consider an example inspired by applications in the verification of asynchronous programs~\cite{GM12,AJKS25}.
In this setting, processes (or threads) can be ``posted'' at runtime to solve a task, and will terminate upon completing the task.
Our example, depicted in \cref{figure:example}, features one clock~$\clock$ per process; symbols~$\ADRactioni{i}$ %
are transition labels of the automaton.
An unbounded number of processes start in the initial location~\locinitexample{}.
In the inner loop, a process can move to location \loclisten{} in order to see whether an input channel carries data.
Once it determines that this is the case (in our example this always happens after some time), it moves to location \locpost{}, which gives the command to post a process that actually reads the data, and then can return to~\locinitexample{}.
In the outer loop, if there is a process that gives the command to read data, \ie{} a process that is in \locpost{}, then another process can accept that command and move to \locreading{}.
After reading for some time, the process will either determine that all the data has been read and move to \locdone{}, or it will timeout and move to \locpost{} to ask another process to carry on reading.
However, this scheme may run into an error if there are processes in \locdone{} and \locreading{} at the same time, modelled by a transition from \locreading{} to \locerror{} that can only be taken if \locdone{} is occupied.
The time to move to \locerror{} is parametric, and should be greater than the (unknown) duration~$\param$.
A natural problem is to identify valuations of~$\param$ for which \locerror{} is unreachable regardless of the number of processes.

\newcommand{\myangle}{28}
\newcommand{\myotherangle}{-42}
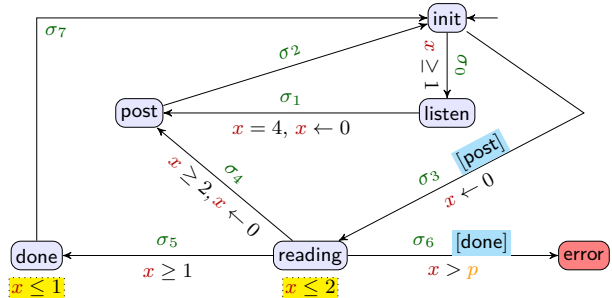
\begin{figure}[tb]
	\centering
	\scalebox{0.9}{
		\begin{tikzpicture}[PTA, font=\footnotesize, xscale=1, yscale=.7]
			\node [] (dummy) at (0,1){};
			\node (qinit) [location, initial right] at (-1,0) {$\locinitexample$};
			\node (qlisten) [location] at (-1,-2) {$\loclisten$};
			\node (qpost) [location] at (-5.5,-2) {$\locpost$};
			\node (qerror) [location, error] at (1,-5) {$\locerror$};
			\node (qreading) [location] at (-3,-5) {$\locreading$};
			\node (qdone) [location] at (-7,-5) {$\locdone$};

			\node[invariant,yshift=-.75em] (doneinvariant) at (qdone.south){$\clockx \leq 1$};
			\node[invariant,yshift=-.75em] (readinginvariant) at (qreading.south){$\clockx \leq 2$};

			\path[edge] (qlisten) edge [] node [above] {$\ADRactioni{1}$} node [below]{$\clockx = 4$, $\clockx \assign 0$} (qpost);

			\draw[edge] (qdone) -- (qdone |- qinit) node[below right]{$\ADRactioni{7}$} -- (qinit);

			\path[edge] (qpost) edge[] node [above,sloped] {$\ADRactioni{2}$} (qinit);

			\path[edge] (qreading) edge[] node [above,,rotate=\myotherangle] {$\ADRactioni{4}$} node [below,,rotate=\myotherangle] {$\clockx \geq 2, \clockx \assign 0$} (qpost);
			\path[edge] (qreading) edge[] node [above,sloped] {$\ADRactioni{5}$} node [below,sloped] {$\clockx \geq 1$} (qdone);
			\path[edge] (qreading) edge[] node [above,xshift=-1.2em]{$\ADRactioni{6}$} node [above,locguard,xshift=1.2em]{\locguardof{\locdone}} node [below,align=center]{$\clockx > \paramp$} (qerror);%

			\draw [edge] (qinit) -- (1,-2) -- node (plop) [above,rotate=\myangle,xshift=-1.2em] {$\ADRactioni{3}$} node[above,rotate=\myangle,locguard,xshift=1.2em]{\locguardof{\locpost}} node[below,rotate=\myangle]{$\clockx \assign 0$} (qreading);

			\path[edge] (qinit) edge[] node [above,sloped] {$\ADRactioni{0}$} node [below,sloped] {$\clockx \geq 1$} (qlisten);

		\end{tikzpicture}

	}
	\caption{Asynchronous data read example (variation from~\cite{AJKS25})}
	\label{figure:example}
\end{figure}

\fakeParagraph{Problems}
We focus here on the \emph{parametrised reachability-emptiness problem}: decide the emptiness of the set of timing parameter valuations for which there exists a number of processes such that a given configuration is reachable.

We consider both \emph{local} properties (reachability condition involving one process),
and \emph{global} properties (which can express the absence of deadlocks, or the fact that \emph{all} processes must reach a given location).

We also distinguish the presence or the absence of invariants---and we will see that this makes a critical difference in decidability.

In this paper, we prove several results regarding parametrised reachability-emptiness:

\begin{itemize}
	\item undecidability of local properties for \PDTNs{} with 1~clock, 1~parameter, with invariants (\cref{ss:undecidability-local-invariants}).
	\item undecidability of global properties for \PDTNs{} with 1~clock, 1~parameter, with or without invariants (\cref{ss:undecidability-global-noinvariants}).
	\item decidability for fully parametric \PDTNs{} with 1~parameter (\cref{ss:fpPDTN});
	\item decidability for \PDTNs{} when parameters are partitioned into lower-bound and upper-bound parameters (\cref{ss:LUPDTNs});
	\item decidability of local properties for \PDTNs{} with 1~clock, arbitrarily many parameters, without invariants (\cref{section:decidability:noinvariants}).
\end{itemize}
The most surprising result emphasizes the high expressive power of invariants (something which has no impact in single PTAs): local properties in \PDTNs{} are decidable without invariants but become undecidable in their presence, for only 1~clock.
In addition, both local and global properties are undecidable with invariants for a single clock and a single parameter---a setting decidable in the context of PTAs taken in isolation: that is to say, while the communication primitive is weak, it is sufficiently expressive to encode models such a 2-counter machines.

Additionally, both as a proof ingredient and as an interesting result \emph{per se}, we show in \cref{section:undecidability:invariants} that the reachability-emptiness problem is undecidable for 1~clock, 1~parameter, with and without invariants and any \emph{a priori} fixed number ($\geq 3$) of processes---even though the parametrised version of this problem, \ie{} for any (non \emph{a priori} fixed) number of processes, is decidable.

\fakeParagraph{Related work}
The concept of identical processes in a timed setting was mainly addressed in networks of processes that either communicate via $k$-wise synchronization~\cite{AJ03,ADM04,ADRST16} or via location guards~\cite{SS20,AEJK24,AJKS25}.
The former model is equivalent to a variant of timed Petri nets~\cite{Ramchandani73,MF76,GMMP91,BS95}, whereas the latter would be equivalent to a form of timed Petri nets restricted to immediate observation steps (as in~\cite{ERW19,RWE20}), which however have not been studied separately to the best of our knowledge.

Very few works study decidability results when combining two types of parameters, \ie{} discrete (number of processes) and continuous (timing parameters).
In~\cite{DG04,LSLD15}, security protocols are studied with unknown timing constants, and an unbounded number of participants.
However, the focus is not on decidability, and the general setting is undecidable.
In~\cite{AKPP16}, action parameters (that can be seen as Boolean variables) and continuous timing parameters are combined (only linearly though) in an extension of PTAs; the mere emptiness of the sets of action and timing parameters for which a location is reachable is undecidable.
In contrast, we exhibit in this work some decidable cases.

The closest work to ours, and presumably the only one to consider timing parameters in the setting of parametrised verification, is in~\cite{ADFL19} where parametric timed broadcast protocols (PTBPs) are introduced.
Our contributions differ in the communication setting: while we consider location guards, \cite{ADFL19} considers broadcast in cliques (in which every message reaches every process) and in reconfigurable topologies (in which the set of receivers is chosen non-deterministically).
Moreover, in this work we study the power of invariants, which are absent from~\cite{ADFL19}.
Our decidability results are also significantly better than in~\cite{ADFL19}:
In~\cite{ADFL19}, parametrised reachability-emptiness (for local properties) is undecidable for PTBPs composed of general PTAs, even with a single clock and without invariants, both in the reconfigurable semantics and in the clique semantics.
The only decidable subcase (for both semantics) for reachability-emptiness in~\cite{ADFL19} is severely restricted with 3 conditions that must hold simultaneously: a single clock per process, parameters partitioned into lower-bound and upper-bound parameters, and bounded (possibly rational-valued) parameters---relaxing any of these conditions leads to undecidability.
In contrast, we prove here decidability for general \PDTNs{} with 1~clock and arbitrarily many parameters, \emph{or} for parameters partitioned into lower-bound and upper-bound.
Also note that our results do not reuse any proof ingredients from~\cite{ADFL19} due to the different communication.

\fakeParagraph{Outline}
We recall the necessary material in \cref{section:PTAs}. %
We formalize parametric disjunctive timed networks in \cref{section:PDTNs}, and our problem in \cref{section:problems}.
We prove undecidability results in \cref{section:undecidability} and decidability results in \cref{section:decidability}.
We conclude in \cref{section:conclusion}.

\section{Parametric timed automata}\label{section:PTAs}

We denote by $\setN, \setNpos, \setZ, \setRgeqzero$ %
the sets of non-negative integers, strictly positive integers, integers, %
and non-negative reals, respectively.

\emph{Clocks} are real-valued variables that all evolve over time at the same rate.
Throughout this paper, we assume a set~$\ClockSet = \{ \clocki{1}, \dots, \clocki{\ClockCard} \} $ of \emph{clocks}.
A \emph{clock valuation} is a function
$\clockval : \ClockSet \to \setRgeqzero$, assigning a non-negative value to each clock.
We write $\ClocksZero$ for the clock valuation assigning $0$ to all clocks.
Given $\resets \subseteq \ClockSet$, we define the \emph{reset} of a valuation~$\clockval$, denoted by $\reset{\clockval}{\resets}$, as follows: $\reset{\clockval}{\resets}(\clock) = 0$ if $\clock \in \resets$, and $\reset{\clockval}{\resets}(\clock)=\clockval(\clock)$ otherwise.
Given a constant $d \in \setRgeqzero$, $\clockval + d$ denotes the valuation \suchthat\ $(\clockval + d)(\clock) = \clockval(\clock) + d$, for all $\clock \in \ClockSet$.

A \emph{(timing) parameter} is an unknown %
integer-valued
constant.
Throughout this paper, we assume a set~$\ParamSet = \{ \parami{1}, \dots, \parami{\ParamCard} \} $ of \emph{parameters}.
A \emph{parameter valuation} $\pval$ is a function $\pval : \ParamSet \to \setN$. %
A \emph{constraint}~$\constraint$ is a conjunction of inequalities over $\ClockSet \cup \ParamSet$ of the form
$\clock \compOp \sum_{1 \leq i \leq \ParamCard} \alpha_i \times \parami{i} + d$, with
$\clock \in \ClockSet$,
${\compOp} \in \{<, \leq, =, \geq, >\}$,
$\parami{i} \in \ParamSet$,
and
$\alpha_i, d \in \setZ$.
We call $d$ a \emph{constant term}.
Given~$\constraint$, we write~$\clockval\models\pval(\constraint)$ if %
the expression obtained by replacing each~$\clock$ with~$\clockval(\clock)$ and each~$\param_i$ with~$\pval(\param_i)$ in~$\constraint$ evaluates to true.
Let $\ConstraintsXP$ denote the set of constraints over~$\ClockSet \cup \ParamSet$.
Let $\CTrue$ denote the constraint made of no inequality, \ie{} representing the whole set of clock and parameter valuations.

\begin{definition}[PTA~\cite{AHV93}]\label{definition:PTA}
	A PTA $\PTA$ is a tuple $\PTA = (\ActionSet, \LocSet, \locinit, %
		\ClockSet, \ParamSet, \invariant, \EdgeSet)$, where:
	\begin{oneenumerate}%
		\item $\ActionSet$ is a finite set of actions;
		\item $\LocSet$ is a finite set of locations;
		\item $\locinit \in \LocSet$ is the initial location;
		\item $\ClockSet$ is a finite set of clocks;
		\item $\ParamSet$ is a finite set of parameters;
		\item $\invariant : \LocSet \to \ConstraintsXP$ is the invariant, assigning to every $\loc \in \LocSet$ a constraint $\invariant(\loc)$ over $\ClockSet \cup \ParamSet$;
		\item $\EdgeSet \subseteq \LocSet \times \ConstraintsXP \times \ActionSet \times 2^{\ClockSet} \times \LocSet$ is a finite set of edges  $\tatran = (\loc,\guard,\action,\resets,\loc')$
		where~$\loc,\loc'\in \LocSet$ are the source and target locations,
			$\guard$ is a constraint (called \emph{guard}),
			$\action \in \ActionSet$,
			and
			$\resets\subseteq \ClockSet$ is a set of clocks to be reset.
	\end{oneenumerate}%
	We say that a location $\loc$ \emph{does not have an invariant} if $\invariant(\loc) = \CTrue$.
\end{definition}
\begin{definition}[Valuation of a PTA]\label{def:valuation-PTA}
	Given a PTA~$\PTA$ and a parameter valuation~$\pval$, we denote by $\valuate{\PTA}{\pval}$ the non-parametric structure where all occurrences of a parameter~$\parami{i}$ have been replaced by~$\pval(\parami{i})$.
	$\valuate{\PTA}{\pval}$ is a \emph{timed automaton}.
\end{definition}

We recall the concrete semantics of a TA using a timed transition system~(TTS).

\begin{definition}[Semantics of a TA]
	Given a PTA $\PTA = (\ActionSet, \LocSet, \locinit, %
		\ClockSet, \ParamSet, \invariant, \EdgeSet)$ and a parameter valuation~$\pval$,
	the semantics of $\valuate{\PTA}{\pval}$ is given by the TTS $\semantics{\valuate{\PTA}{\pval}} = (\StateSet, \concstateinit, \ActionSet \cup \setRgeqzero, \transition)$, with
	\begin{enumerate}
		\item $\StateSet = \big\{ (\loc, \clockval) \in \LocSet \times \setRgeqzero^\ClockCard \mid \clockval \models \valuate{\invariant(\loc)}{\pval} \big\}$,
		 $\concstateinit = (\locinit, \ClocksZero) $,
		\item  $\transition$ consists of the discrete and (continuous) delay transition relations:
		\begin{enumerate}
			\item discrete transitions: $(\loc,\clockval) \transitionWith{\tatran} (\loc',\clockval')$,
			if $(\loc, \clockval) , (\loc',\clockval') \in \StateSet$, and there exists $\tatran = (\loc,\guard,\action,\resets,\loc') \in \EdgeSet$, such that $\clockval'= \reset{\clockval}{\resets}$, and $\clockval\models\pval(\guard$).
			\item delay transitions: $(\loc,\clockval) \transitionWith{\cDelay} (\loc, \clockval + \cDelay)$, with $\cDelay \in \setRgeqzero$, if $\forall \cDelay' \in [0, \cDelay], (\loc, \clockval + \cDelay') \in \StateSet$.
		\end{enumerate}
	\end{enumerate}
\end{definition}

Moreover we write $(\loc, \clockval)\longuefleche{(\cDelay, \tatran)} (\loc',\clockval')$ for a combination of a delay and a discrete transition if
$\exists  \clockval'' :  (\loc,\clockval) \transitionWith{\cDelay} (\loc,\clockval'') \transitionWith{\tatran} (\loc',\clockval')$.

Given a TA~$\TA$ with concrete semantics $\semantics{\TA}$, we refer to the states of~$\StateSet$ as the \emph{concrete states} of~$\TA$.
A \emph{run} of~$\TA$ is an alternating sequence of concrete states of~$\TA$ and pairs of delays and edges starting from the initial state $\concstateinit$ of the form
$(\loci{0}, \clockval_{0}), (d_0, \tatran_0), (\loci{1}, \clockval_{1}), \cdots$
with
$i = 0, 1, \dots$, $\tatran_i \in \EdgeSet$, $d_i \in \setRgeqzero$ and
$(\loci{i}, \clockval_{i}) \longuefleche{(d_i, \tatran_i)} (\loci{i+1}, \clockval_{i+1})$.
Given a TA~$\TA$, we say that a location~$\loc$ is \emph{reachable} if there exists a state $(\loc, \clockval)$ that appears on a run of~$\TA$.

\fakeParagraph{Reachability-emptiness}
Given a PTA~$\PTA$ and a location~$\loc$, the \emph{reachability-emptiness problem} asks whether the set of parameter valuations~$\pval$ such that $\loc$ is reachable in $\valuate{\PTA}{\pval}$ is empty.

\begin{definition}[L/U-PTA~\cite{HRSV02}]\label{definition:LUPTA}
	An \emph{L/U-PTA} (lower-bound/upper-bound PTA) is a PTA where~$\ParamSet$ is partitioned into $\ParamSet = \ParamSetL \uplus \ParamSetU$, where $\ParamSetL$ (resp.\ $\ParamSetU$) denotes lower-bound (resp.\ upper-bound) parameters,
	so that each lower-bound (resp.\ upper-bound) parameter~$\parami{i}$ must be such that,
	for every constraint $\clock \compOp \sum_{1 \leq i \leq \ParamCard} \alpha_i \times \parami{i} + d$, we have:
	\begin{oneenumerate}
		\item ${\compOp} \in \set{ \leq, < }$ implies $\alpha_i \leq 0$ (resp.\ $\alpha_i \geq 0$), and
		\item ${\compOp} \in \set{ \geq, > }$ implies $\alpha_i \geq 0$ (resp.\ $\alpha_i \leq 0$).
	\end{oneenumerate}
\end{definition}
A PTA is \emph{fully parametric} whenever it has no constant term (apart from~0):

\begin{definition}[Fully parametric PTA~{\cite[Definition~4.6]{HRSV02}}]\label{definition:fpPTA}
	A PTA~$\PTA$ is \emph{fully parametric} if %
		every constraint in~$\PTA$ is of the form $\clock \compOp \sum_{1 \leq i \leq \ParamCard} \alpha_i \times \parami{i}$, with
	$\parami{i} \in \ParamSet$ and $\alpha_i \in \setZ$.
\end{definition}
Our subsequent undecidability proofs work by reduction from
the halting problem for 2-counter machines.
We borrow the following material from~\cite{ALR22}.
A deterministic 2-counter machine (``2CM'')~\cite{Minsky67} has two non-negative counters $\Counter_1$ and~$\Counter_2$, a finite number of states and a finite number of transitions, which can be of the form (for $\counterGenericIndex \in \{ 1, 2\}$):
	\begin{ienumerate}
		\item ``when in state~$\cms_i$, increment~$\Counter_\counterGenericIndex$ and go to~$\cms_j$''; or
		\item ``when in state~$\cms_i$, if $\Counter_\counterGenericIndex = 0$ then go to~$\cms_k$, otherwise decrement $\Counter_\counterGenericIndex$ and go to~$\cms_j$''.
	\end{ienumerate}

The 2CM starts in state~$\cms_0$ with the counters set to~0.
The machine follows a deterministic transition function, meaning for each combination of state and counter conditions, there is exactly one action to take.
The \emph{halting problem} consists in deciding whether some distinguished state called \cmshalt{} can be reached or not.
This problem is known to be undecidable~\cite{Minsky67}.
\section{Parametric disjunctive timed networks}\label{section:PDTNs}
We extend guarded TAs defined in~\cite{AEJK24} with timing parameters as in PTAs.
We follow the terminology and use abbreviations from~\cite{SS20,AEJK24}.

\begin{definition}[Guarded Parametric Timed Automaton (\gPTA)]\label{definition:gPTA}
	A \gPTA{} $\PTA$ is a tuple $\PTA = (\ActionSet, \LocSet, \locinit, %
		\ClockSet, \ParamSet, \invariant, \gEdgeSet)$, where:
	\begin{oneenumerate}%
		\item $\ActionSet$ is a finite set of actions;
		\item $\LocSet$ is a finite set of locations;
		\item $\locinit \in \LocSet$ is the initial location;
		\item $\ClockSet$ is a finite set of clocks;
		\item $\ParamSet$ is a finite set of parameters;
		\item $\invariant : \LocSet \to \ConstraintsXP$ is the invariant, assigning to every $\loc \in \LocSet$ a constraint $\invariant(\loc)$ (called \emph{invariant});
		\item $\gEdgeSet \subseteq \LocSet \times \ConstraintsXP \times \left( \LocSet \cup \{\top\} \right) \times \ActionSet \times 2^{\ClockSet} \times \LocSet$ is a finite set of edges $\gtatran = (\loc, \guard, \locguard, \action, \resets, \loc')$
		where~$\loc,\loc'\in \LocSet$ are the source and target locations,
		$\guard$ is a constraint (called \emph{guard}),
		$\locguard$ is the location guard,
		$\action \in \ActionSet$,
		and
		$\resets\subseteq \ClockSet$ is a set of clocks to be reset.
	\end{oneenumerate}%
\end{definition}

Intuitively, an edge $\gtatran = (\loc, \guard, \locguard, \action, \resets, \loc') \in \gEdgeSet$ takes the automaton from location $\loc$ to~$\loc'$;
$\gtatran$~can only be taken if \emph{guard}~$\guard$ and \emph{location guard}~$\locguard$ are both satisfied, and it resets all clocks in~$\resets$.
Note that satisfaction of location guards is only meaningful in a \emph{network} of~\gPTAs{} (defined below).
Intuitively, a location guard~$\locguard$ is satisfied if it is $\top$ or if another automaton in the network currently occupies location~$\locguard$.
\begin{example}
	In \cref{figure:example}, the transition to~\locerror{} is guarded both by a guard $\clock > \param$ and by a location guard $\locguardof{\locdone}$.
	In contrast, the location guard to \locdone{} is~$\top$ (and omitted in the figure), \ie{} it can be taken without assumption on the location of other processes.
\end{example}

A \gPTA{} is an \emph{L/U-\gPTA{}} if parameters are partitioned into lower-bound and upper-bound parameters (as in \cref{definition:LUPTA}).
A \gPTA{} is \emph{fully parametric} whenever it has no constant term (apart from~0) in its guards and invariants, similarly to \cref{definition:fpPTA}.

\fakeParagraph{Recalling the semantics of \NTAs{}}
A \gPTA{} with $\ParamSet = \emptyset$ is called a guarded timed automaton (\gTA)~\cite{AEJK24}.
Given a \gPTA{} $\PTA$ and a parameter valuation~$\pval$, we denote by $\valuate{\PTA}{\pval}$ the non-parametric structure where all occurrences of a parameter~$\parami{i}$ have been replaced by~$\pval(\parami{i})$; $\valuate{\PTA}{\pval}$ is a \gTA{}.

Let $\TA$ be a~\gTA{}.
We denote by $\network{\TA}{n}$ the parallel composition $\TA \parallel \cdots \parallel \TA$ of $n$~copies of~$\TA$, also called a \emph{network of timed automata}~(NTA) of size~$n$.
Each copy of~$\TA$ in the NTA~$\network{\TA}{n}$ is called a \emph{process}.
A \emph{configuration}
$\nConfig$ of an NTA $\network{\TA}{n}$ is a tuple $\nConfig = \dtnconfig$, where every $(\loc_i,\clockval_i)$ is a concrete state of~$\TA$.
The semantics of $\network{\TA}{n}$ can be defined as a TTS $(\ConfigsSet, \nConfigInit,T)$, where $\ConfigsSet$
denotes the set of all configurations of~$\network{\TA}{n}$, $\nConfigInit$ is the unique initial configuration $(\locinit,\mathbf{0})^n$, and the transition relation $T$ is the union of the following delay and discrete transitions:

\begin{description}\label{semantics-NTAs}
	\item[delay transition] $ \transitionOf {\dtnconfig} {\cDelay} {\dtnconfigdelay}$, with $\cDelay \in \setRgeqzero$,
	if $\forall i \in \{1, \dots, n\} , \forall \cDelay' \in [0, \cDelay] : \clockval_i + \cDelay' \models \invariant(\loc_i)$, \ie{} we can delay $\cDelay \in \setRgeqzero $ units of time if all clock invariants are satisfied until the end of the delay.

	\item[discrete transition] $\transitionOf{\dtnconfig } {(i, \action)} {\dtnconfigsucc }$ for some $i \in \{1, \dots, n\}$ if
		\begin{oneenumerate}%
			\item $\transitionOf{(\loc_i,\clockval_i)} {\gtatran} {(\loc_i',\clockval_i')}$ is a discrete transition%
				\footnote{%
				Strictly speaking, $\transitionOf{(\loc_i,\clockval_i)} {\gtatran} {(\loc_i',\clockval_i')}$ is a transition of the TA obtained from~$\TA$ by replacing location guards with~$\top$.
			}
			of~$\TA$
			with
			$\gtatran = (\loc_i, \guard, \locguard, \action, \resets, \loc_i')$ for some $\guard$, $\locguard$ and~$\resets$,
			\item $\locguard = \top$ or $\loc_j = \locguard$ for some $j \in \{1, \dots, n\} \setminus \{i\}$, and
			\item $\loc_j' = \loc_j$ and $\clockval_j' = \clockval_j$ for all $j \in \{1, \dots, n\} \setminus \{i\}$.
		\end{oneenumerate}%
\end{description}

That is, location guards $\locguard$ are interpreted as disjunctive guards: unless $\locguard = \top$, at least one of the other processes needs to occupy location~$\locguard$ in order for process~$i$ to pass this guard.

We write $\transitionOf{\nConfig}{\cDelay, (i,\action)}{\nConfig''}$ for a delay transition $\transitionOf{\nConfig}{\cDelay}{ \nConfig'}$ followed by a discrete transition $\transitionOf{\nConfig'}{(i,\action)}{\nConfig''}$.
Then, a \emph{timed path} of~$\network{\TA}{n}$ is a finite sequence $\gcomputation = \nConfig_0 \xrightarrow{\cDelay_0, (i_0, \action_0)} \cdots \xrightarrow{\cDelay_{l-1}, (i_{l-1}, \action_{l-1})} \nConfig_l$.
A timed path $\gcomputation$ of~$\network{\TA}{n}$ is a \emph{computation} if $\nConfig_0 = \nConfigInit$.
We say that $\nConfig_l$ is \emph{reachable} in~$\network{\TA}{n}$.

We write $\loc \in \nConfig$ if $\nConfig = \big((\loc_1,\clockval_1),\ldots,(\loc_n,\clockval_n)\big)$ and $\loc = \loc_i$ for some $i \in \{1, \dots, n\}$.
We say that a location $\loc$ is \emph{reachable} in~$\network{\TA}{n}$ if there exists a reachable configuration $\nConfig$ s.t.\ $\loc \in \nConfig$.

Given a \gPTA{}~$\PTA$, we denote by $\Pnetwork{\PTA}{n}$ the parallel composition $\PTA \parallel \cdots \parallel \PTA$ of $n$~copies of~$\PTA$, also called a \emph{network of parametric timed automata}~(\NPTA{}).
Given an \NPTA{} $\Pnetwork{\PTA}{n}$ and a parameter valuation~$\pval$, we denote by $\valuate{\Pnetwork{\PTA}{n}}{\pval}$ the non-parametric structure where all occurrences of a parameter~$\parami{i}$ have been replaced by~$\pval(\parami{i})$;
note that $\valuate{\Pnetwork{\PTA}{n}}{\pval}$ is an NTA.

\begin{definition}
	A given \gPTA{}~$\PTA$ induces a \emph{parametric disjunctive timed network (\PDTN{})}
	$\Pnetwork{\PTA}{\infty}$, defined as the following family of~\NPTAs{}:
	$\Pnetwork{\PTA}{\infty} = \{\Pnetwork{\PTA}{n} \mid n \in \setNpos\}$.
\end{definition}

Given a parametric disjunctive timed network $\Pnetwork{\PTA}{\infty}$ and a parameter valuation~$\pval$, $\Pnetwork{(\valuate{\PTA}{\pval})}{\infty}$ is a \emph{disjunctive timed network} (\DTN{})~\cite{SS20}.

Given a location~$\loc$ of a \gPTA{}~$\PTA$, given a parameter valuation~$\pval$,
	we say that $\loc$ is reachable in the \DTN{} $\Pnetwork{(\valuate{\PTA}{\pval})}{\infty}$ if there exists~$n \in \setNpos$ such that $\loc$ is reachable in $\Pnetwork{(\valuate{\PTA}{\pval})}{n}$.

\fakeParagraph{Subclasses of \PDTNs{}}
A \PDTN{} $\Pnetwork{\PTA}{\infty}$ induced by a \gPTA{}~$\PTA$ is an \emph{L/U-\PDTN{}} if $\PTA$ is an L/U-\gPTA{}.
Similarly, $\Pnetwork{\PTA}{\infty}$ is a \emph{fully parametric \PDTN{}} if $\PTA$ is a fully parametric \gPTA{}.

\section{Problems for parametric disjunctive timed networks}\label{section:problems}

\fakeParagraph{Reachability}
In~\cite{AEJK24},
the parametrised reachability problem (PR)
consists in deciding whether, given a \gTA{}~$\TA$ and a location~$\loc$, there exists $n \in \setN$ such that $\loc$ is reachable in~$\network{\TA}{n}$.
Here, we consider a parametric version.
The emptiness extension consists in asking whether the set of timing parameter valuations for which PR holds is empty.

\defProblem
	{Parameterized reachability-emptiness problem (PR-e)}
	{a \gPTA{} $\PTA$ and a location~$\loc$}
	{Decide the emptiness of the set of timing parameter valuations~$\pval$ for which $\loc$ is reachable in~$\Pnetwork{(\valuate{\PTA}{\pval})}{\infty}$.
	}
\begin{example}
	In \cref{figure:example}, PR-e with \locerror{} as target location does \emph{not} hold: for $\pval(\param) = 1$, location \locerror{} can be reached for $\geq 3$ processes.
\end{example}

Without invariants, $|\LocSet|$ is a \emph{cutoff} (\ie{} a number of processes above which the reachability is homogeneous---better cutoffs are known for local properties).
Intuitively, this is because, without invariants, a single process that reaches such a location can stay there and enable the guard forever.
However, with invariants, such a simple cutoff does not work (even in the absence of timing parameters), since invariants might force a process to leave a location and lots of different processes might be needed to occupy a location at different points in time.

\fakeParagraph{Global reachability}\label{paragraph:definition:global}
Global properties refer to the numbers of processes in given locations; we express these using constraints.
Formally, a \emph{global reachability property} is defined by a constraint $\varphi$ in the following grammar:
\( \varphi ::= \#\loc \geq 1 \mid \#\loc = 0  \mid \varphi \land \varphi \mid \varphi \lor \varphi\),
where
$\loc \in \LocSet$ and $\#\loc$ refers to the number of processes in $\loc$.
Note that the ``$\#\loc = 0$'' term is responsible for the ``global'' nature of such properties, \ie{} it must hold for \emph{all} processes that they are not in~$\loc$.
The satisfaction of a constraint $\varphi$ by $\nConfig=\big((\loc_1,\clockval_1),\ldots,(\loc_n,\clockval_n)\big)$ is defined naturally:
$\nConfig \models \#\loc \geq 1$ if $\loc \in \nConfig$;
$\nConfig \models \#\loc = 0$ if $\loc \notin \nConfig$;
and as usual for Boolean combinations.
The parametrised global reachability problem (PGR) consists in deciding whether, given a \gTA{}~$\TA$ and a global reachability property $\varphi$, there exists $n \in \setN$ such that a configuration $\nConfig$ with $\nConfig \models \varphi$ is reachable in~$\network{\TA}{n}$.
Again, we extend this problem to the emptiness of the set of timing parameters:

\defProblem
	{Parameterized global reachability-emptiness problem (PGR-e)}
	{a \gPTA{} $\PTA$ and a global reachability property $\varphi$}
	{Decide the emptiness of the set of timing parameter valuations~$\pval$ for which
		a configuration $\nConfig$ with $\nConfig \models \varphi$ is reachable in~$\Pnetwork{(\valuate{\PTA}{\pval})}{\infty}$.
	}

As special cases, the parametrised global reachability problem includes \emph{control-state reachability} (where $\varphi$ is $\#\loc \geq 1$ for a single location~$\loc$; also expressible as a PR-e problem) and \emph{detection of timelocks that are due to location guards} (where $\varphi$ is a disjunction over conjunctions of the form $\#\loc \geq 1 \land \#\loc_1 = 0 \land \cdots \land \#\loc_m = 0$, where $\loc$ can only be left through edges guarded by one of the $\loc_1, \ldots, \loc_m$).
We will be particularly interested in the global property asking whether \emph{all} processes reach a given final configuration
(where $\varphi$ is of the form $\#\loc_1 = 0 \land \cdots \land \#\loc_m = 0$ with $m = |\LocSet|-1$),
sometimes called the \emph{\textsc{target} problem}~\cite{DSZ10}.

\begin{example}
	In \cref{figure:example}, consider the global property~$\varphi$ stating that all processes must be in~\locerror{}.
	This property can be written as $\#\locinitexample = 0 \land \#\loclisten = 0 \land \#\locpost = 0 \land \#\locreading = 0 \land \#\locdone = 0$;
	then, PGR-e holds: no parameter valuation can allow all processes to be simultaneously in~\locerror{}, whatever the number of processes (this comes from the fact that the transition to \locerror{} is guarded by \locdone{}, so at least one process must be elsewhere).
\end{example}
\section{Undecidability results}\label{section:undecidability}
\subsection{Fixed number of processes}\label{section:undecidability:invariants}

We first consider a fixed number of processes; the problem addressed here is therefore not (yet) parametrised reachability-emptiness, but only reachability-emptiness.
These results will be used as important proof ingredients for the results in \cref{ss:undecidability-local-invariants,ss:undecidability-global-noinvariants}.
(Then, the three results in \cref{section:decidability} use 3 different proof techniques, completely different from this one.)

\subsubsection{Undecidability for 3 processes}

In the following, we show that the emptiness of the parameter valuations set for which a location is reachable in an \NPTA{} made of exactly 3~processes (``reachability-emptiness'') is undecidable.
We first prove the result with invariants (\cref{proposition:undecidability:3-processes-invariants}), and then show it also holds without (\cref{proposition:undecidability:3-processes-noinvariants}).
The idea is that invariants are not needed for 3~processes, but will be necessary when proving undecidability for an unbounded number of processes (\cref{theorem:undecidability-local-invariants}).

Before that, we first reprove a well-known result stating that reachability-emptiness is undecidable for PTAs with 3 clocks and a single parameter.
This result was already proved (with 3 or more clocks) in, \eg{}~\cite{AHV93,BBLS15,ALM20} with various proof flavors.
The construction we introduce in the proof of \cref{lemma:undecidability:reach-emptiness:integers} will be then reused and modified in the proof of \cref{proposition:undecidability:3-processes-invariants}, which is why we give it first with full details.

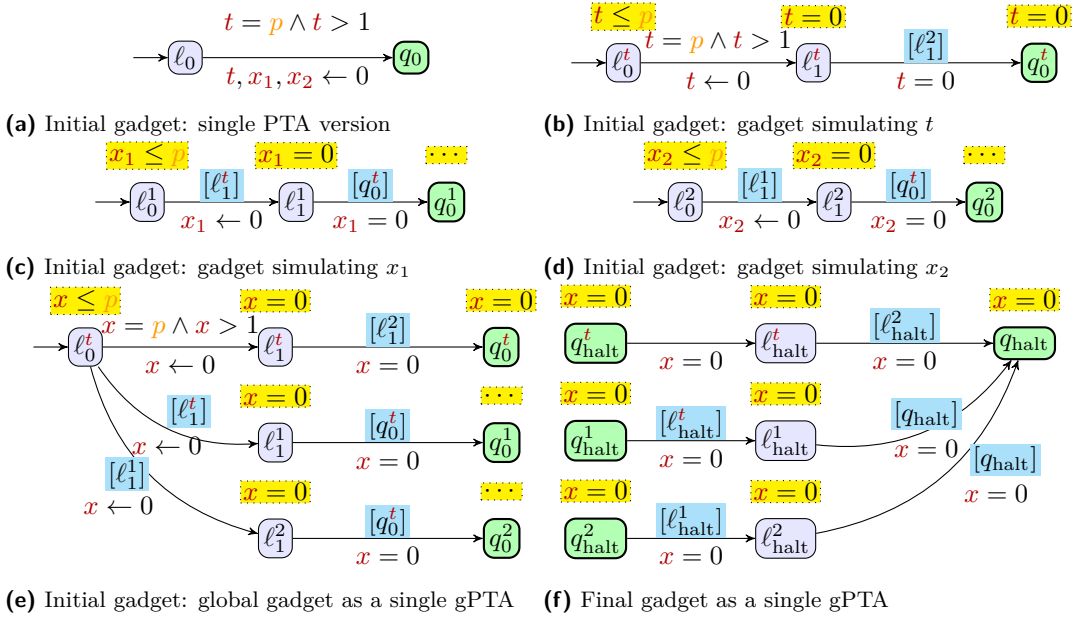
\begin{figure*}
	\begin{subfigure}[b]{.5\textwidth}
	\centering
	\scalebox{\generalFigsScaleFactor}{
			\begin{tikzpicture}[PTA, node distance=2.5cm]
			\node[location, initial] (l0) {$\loci{0}$};
			\node[location2CM, right=of l0] (s0) {$\cmspta_0$};

			\path (l0) edge node[above, yshift=.5em, align=center]{$\clockCMt = \paramp \land \clockCMt > 1$} node[below]{$\clockCMt, \clockCMcone, \clockCMctwo \assign 0$} (s0);
			\end{tikzpicture}
	}

	\caption{Initial gadget: single PTA version}
	\label{figure:2CM:original:initial}
	\end{subfigure}
	\begin{subfigure}[b]{.5\textwidth}
	\centering
	\scalebox{\generalFigsScaleFactor}{
		\begin{tikzpicture}[PTA, node distance=2.5cm]
			\node[location, initial] (l0) {$\locij{0}{\clockCMt}$};
			\node[location, right of=l0] (l1) {$\locij{1}{\clockCMt}$};
			\node[location2CM, right=of l1] (s0) {$\cmspta_0^{\clockCMt}$};

			\node[invariantNorth] at (l0.north) {$\clockCMt \leq \paramp$};
			\node[invariantNorth] at (l1.north) {$\clockCMt = 0$};
			\node[invariantNorth] at (s0.north) {$\clockCMt = 0$};

			\path (l0) edge node[above, align=center]{$\clockCMt = \paramp \land \clockCMt > 1$} node [below]{$\clockCMt \assign 0$} (l1);
			\path (l1) edge node[above, locguard]{$\locguardof{\locij{1}{2}}$} node[below]{$\clockCMt = 0$}  (s0);
		\end{tikzpicture}
	}

	\caption{Initial gadget: gadget simulating~$\clockCMt$}
	\label{figure:2CM:initial:t}
	\end{subfigure}
	\begin{subfigure}[b]{.5\textwidth}
	\centering
	\scalebox{\generalFigsScaleFactor}{
		\begin{tikzpicture}[PTA, node distance=1.5cm]
			\node[location, initial] (l0) {$\locij{0}{1}$};
			\node[location, right=of l0] (l1) {$\locij{1}{1}$};
			\node[location2CM, right=of l1] (s0) {$\cmspta_0^1$};

			\node[invariantNorth] at (l0.north) {$\clockCMcone \leq \paramp$};
			\node[invariantNorth] at (l1.north) {$\clockCMcone = 0$};
			\node[invariantNorth] at (s0.north) {$\cdots$};

			\path (l0) edge node[above, locguard]{$\locguardof{\locij{1}{\clockCMt}}$} node [below]{$\clockCMcone \assign 0$} (l1);
			\path (l1) edge node[above, locguard]{$\locguardof{\cmspta_0^{\clockCMt}}$} node[below]{$\clockCMcone = 0$} (s0);
		\end{tikzpicture}
	}

	\caption{Initial gadget: gadget simulating~$\clockCMcone$}
	\label{figure:2CM:initial:x1}
	\end{subfigure}
	\begin{subfigure}[b]{.5\textwidth}
	\centering
	\scalebox{\generalFigsScaleFactor}{
	\begin{tikzpicture}[PTA, node distance=1.5cm]
			\node[location, initial] (l0) {$\locij{0}{2}$};
			\node[location, right=of l0] (l1) {$\locij{1}{2}$};
			\node[location2CM, right=of l1] (s0) {$\cmspta_0^2$};

			\node[invariantNorth] at (l0.north) {$\clockCMctwo \leq \paramp$};
			\node[invariantNorth] at (l1.north) {$\clockCMctwo = 0$};
			\node[invariantNorth] at (s0.north) {$\cdots$};

			\path (l0) edge node[above, locguard]{$\locguardof{\locij{1}{1}}$} node [below]{$\clockCMctwo \assign 0$} (l1);
			\path (l1) edge node[above, locguard]{$\locguardof{\cmspta_0^{\clockCMt}}$} node[below]{$\clockCMctwo = 0$} (s0);
	\end{tikzpicture}
	}

	\caption{Initial gadget: gadget simulating~$\clockCMctwo$}
	\label{figure:2CM:initial:x2}
	\end{subfigure}
	\begin{subfigure}[b]{.5\textwidth}
	\centering
	\scalebox{\generalFigsScaleFactor}{
		\begin{tikzpicture}[PTA, node distance=2.5cm]
			\node[location, initial] at (0,0) (l0) {$\locij{0}{\clockCMt}$};
			\node[location, right of=l0] (l1) {$\locij{1}{\clockCMt}$};
			\node[location2CM, right=of l1] (s0) {$\cmspta_0^{\clockCMt}$};
			\node[invariantNorth] at (l0.north) {$\clockx \leq \paramp$};
			\node[invariantNorth] at (l1.north) {$\clockx = 0$};
			\node[invariantNorth] at (s0.north) {$\clockx = 0$};

			\node[location, below=of l1, yshift=5em] (l1_1) {$\locij{1}{1}$};
			\node[location2CM, right=of l1_1] (s0_1) {$\cmspta_0^1$};
			\node[invariantNorth] at (l1_1.north) {$\clockx = 0$};
			\node[invariantNorth] at (s0_1.north) {$\cdots$};

			\node[location, below=of l1_1, yshift=5em] (l1_2) {$\locij{1}{2}$};
			\node[location2CM, right=of l1_2] (s0_2) {$\cmspta_0^2$};
			\node[invariantNorth] at (l1_2.north) {$\clockx = 0$};
			\node[invariantNorth] at (s0_2.north) {$\cdots$};

			\path (l0) edge node[above, align=center]{$\clockx = \paramp \land \clockx > 1$} node [below]{$\clockx \assign 0$} (l1);
			\path (l1) edge node[above, locguard]{$\locguardof{\locij{1}{2}}$} node[below]{$\clockx = 0$}  (s0);

			\path (l0) edge[bend right] node[locguard]{$\locguardof{\locij{1}{\clockCMt}}$} node[below]{$\clockx \assign 0$} (l1_1);
			\path (l1_1) edge node[above, locguard]{$\locguardof{\cmspta_0^{\clockCMt}}$} node[below]{$\clockx = 0$} (s0_1);

			\path (l0) edge[bend right] node[locguard, left]{$\locguardof{\locij{1}{1}}$} node[below left,yshift=-.5em, xshift=.5em]{$\clockx \assign 0$} (l1_2);
			\path (l1_2) edge node[above, locguard]{$\locguardof{\cmspta_0^{\clockCMt}}$} node[below]{$\clockx = 0$} (s0_2);
		\end{tikzpicture}
	}

	\caption{Initial gadget: global gadget as a single \gPTA{}}
	\label{figure:2CM:initial:full}
	\end{subfigure}
	\begin{subfigure}[b]{.5\textwidth}
	\centering
	\scalebox{\generalFigsScaleFactor}{
		\begin{tikzpicture}[PTA, node distance=2.5cm]
			\node[location2CM] (l0_t) {$\cmspta_{\cmshaltname}^{\clockCMt}$};
			\node[location, right of=l0_t] (l1_t) {$\locij{\cmshaltname}{\clockCMt}$};
			\node[location2CM, right=of l1] (qhalt) {$\lochalt$};
			\node[invariantNorth] at (l0_t.north) {$\clockx = 0$};
			\node[invariantNorth] at (l1_t.north) {$\clockx = 0$};
			\node[invariantNorth] at (qhalt.north) {$\clockx = 0$};

			\node[location2CM, below=of l0_t, yshift=5em] (l0_1) {$\cmspta_{\cmshaltname}^{1}$};
			\node[location, right of=l0_1] (l1_1) {$\locij{\cmshaltname}{1}$};
			\node[invariantNorth] at (l0_1.north) {$\clockx = 0$};
			\node[invariantNorth] at (l1_1.north) {$\clockx = 0$};

			\node[location2CM, below=of l0_1, yshift=5em] (l0_2) {$\cmspta_{\cmshaltname}^{2}$};
			\node[location, right of=l0_2] (l1_2) {$\locij{\cmshaltname}{2}$};
			\node[invariantNorth] at (l0_2.north) {$\clockx = 0$};
			\node[invariantNorth] at (l1_2.north) {$\clockx = 0$};

			\path (l0_t) edge %
				node[below]{$\clockx = 0$} (l1_t);
			\path (l1_t) edge node[above, locguard]{$\locguardof{\locij{\cmshaltname}{2}}$} node[below]{$\clockx = 0$} (qhalt);

			\path (l0_1) edge node[above, locguard]{$\locguardof{\locij{\cmshaltname}{\clockCMt}}$} node[below]{$\clockx = 0$} (l1_1);
			\path (l1_1) edge[bend right] node[above, locguard]{$\locguardof{\lochalt}$} node[below]{$\clockx = 0$} (qhalt);

			\path (l0_2) edge node[above, locguard]{$\locguardof{\locij{\cmshaltname}{1}}$} node[below]{$\clockx = 0$} (l1_2);
			\path (l1_2) edge[bend right] node[above right, locguard, xshift=1em]{$\locguardof{\lochalt}$} node[below,xshift=2em]{$\clockx = 0$} (qhalt);
		\end{tikzpicture}
	}

	\caption{Final gadget as a single \gPTA{}}
	\label{figure:2CM:final}
	\end{subfigure}

	\caption{Initial and final gadgets}
	\label{figure:2CM:initial}
\end{figure*}
\begin{lemma}\label{lemma:undecidability:reach-emptiness:integers}
	Reachability-emptiness is undecidable for PTAs with 3~clocks and 1~parameter.
\end{lemma}
\begin{proof}
	We reduce from the halting problem of 2-counter machines, which is undecidable~\cite{Minsky67}.
	Given a 2CM~$\calM$, we encode it as a PTA~$\PTA$.
   	Let us describe this encoding in detail, as we will modify it in the subsequent proofs.

	Each state $\cms_i$ of the machine is encoded as a location of the PTA, which we call~$\cmspta_i$.
	The counters are encoded using clocks $\clockCMt$, $\clockCMcone$ and $\clockCMctwo$ and one integer-valued parameter~$\param$, with the following relations with the values $\countervalue_1$ and~$\countervalue_2$ of counters $\Counter_1$ and~$\Counter_2$: when $\clockCMt = 0$, we have $\clockCMcone = \countervalue_1$ and $\clockCMctwo = \countervalue_2$.
	This clock encoding is classical for integer-valued parameters, \eg{} \cite{BBLS15,ALM20}.
	The parameter typically encodes the maximal value of the counters along a run.

	We initialize the clocks with the gadget in \cref{figure:2CM:original:initial} (that also blocks the case where $\param \leq 1$).
	Note that, throughout the paper, we highlight in thick green style the locations of the PTA corresponding to a state of the 2CM (in contrast with other locations added in the encoding to maintain the matching between the clock values and the counter values).
	Since all clocks are initially~0, in \cref{figure:2CM:original:initial} clearly, when in $\cmspta_0$ with $\clockCMt = 0$, we have $\clockCMcone = \clockCMctwo = 0$, which indeed corresponds to counter values~$0$.

	We now present the gadget encoding the increment instruction of~$\Counter_1$ in \cref{figure:2CM:original:increment}.
	The edge from $\cmspta_i$ to $\loci{i1}$ only serves to clearly indicate the entry in the increment gadget and is done in $0$ time unit.
	Since every edge is guarded by one equality, there are really only two timed paths that go through the gadget:
	one going through $\loci{i2}$ and one through $\loci{i2}'$, depending on the respective order between $\countervalue_1$ and~$\countervalue_2$.
	Observe that on both timed paths the gadget lasts exactly $\param$ time units (due to the guards and resets of~$\clockCMt$).
	In addition, $\clockCMctwo$ is reset exactly when it equals~$\param$, hence its value when entering the gadget is identical to its value when reaching~$\cmspta_j$.
	Therefore $\countervalue_2$ is unchanged.
	Now, $\clockCMcone$ is reset when it equals~$\param - 1$, hence its value after the gadget of duration~$\param$ is incremented by~1 compared to its value when entering the gadget.
	Therefore $\countervalue_1$ is incremented by~1 when reaching~$\cmspta_j$, as expected.

\begin{figure*}
	\begin{subfigure}[b]{.5\textwidth}
	\centering
	\scalebox{\generalFigsScaleFactor}{
	\begin{tikzpicture}[PTA, node distance=1.5cm]
		\node[location2CM] (si) {$\cmspta_i$};
		\node[location, right=of si] (l1) {$\loci{i1}$};
		\node[location, above right=of l1] (l2) {$\loci{i2}$};
		\node[location, below right=of l1] (l2p) {$\loci{i2}'$};
		\node[location, below right=of l2] (l3) {$\loci{i3}$};
		\node[location2CM, right=of l3] (sj) {$\cmspta_j$};

		\path[edge] (si) edge[] node {$\clockCMt = 0$} (l1);
		\path[edge] (l1) edge[] node[align=center] {$\clockCMctwo = \paramp$\\$\clockCMctwo \assign 0$} (l2);
		\path[edge] (l1) edge[] node[align=center,swap] {$\clockCMcone = \paramp - 1$\\ $\clockCMcone \assign 0$} (l2p);
		\path[edge] (l2) edge[] node[align=center] {$\clockCMcone = \paramp - 1$\\ $ \clockCMcone \assign 0$} (l3);
		\path[edge] (l2p) edge[] node[align=center,swap] {$\clockCMctwo = \paramp$\\ $ \clockCMctwo \assign 0$} (l3);
		\path[edge] (l3) edge[] node[align=center] {$\clockCMt= \paramp$\\ $ \clockCMt\assign 0$} (sj);
	\end{tikzpicture}
	}

	\caption{Increment~$\Counter_1$: single PTA version}
	\label{figure:2CM:original:increment}
	\end{subfigure}
	\begin{subfigure}[b]{.5\textwidth}
	\centering
	\scalebox{\generalFigsScaleFactor}{
	\begin{tikzpicture}[PTA, node distance=1.5cm]
		\node[location2CM] (si) {$\cmspta_i^{\clockCMt}$};
		\node[location, right=of si] (li) {$\locij{i}{\clockCMt}$};
		\node[location2CM, right=of li] (sj) {$\cmspta_j^{\clockCMt}$};

		\node[invariantNorth] at (si.north) {$\clockCMt = 0$};
		\node[invariantNorth] at (li.north) {$\clockCMt \leq \paramp$};
		\node[invariantNorth] at (sj.north) {$\clockCMt = 0$};

		\path[edge] (si) edge[] node {$\clockCMt = 0$} (li);
		\path[edge] (li) edge[] node[above] {$\clockCMt = \paramp$ } node[below]{ $ \clockCMt \assign 0$} (sj);
	\end{tikzpicture}
	}

	\caption{Increment~$\Counter_1$: gadget simulating~$\clockCMt$}
	\label{figure:2CM:increment:t}
	\end{subfigure}
	\begin{subfigure}[b]{.5\textwidth}
	\centering
	\scalebox{\generalFigsScaleFactor}{
	\begin{tikzpicture}[PTA, node distance=2.5cm]
		\node[location2CM] (si) {$\cmspta_i^1$};
		\node[location, xshift=-2em, right=of si] (li) {$\locij{i}{1}$};
		\node[location2CM, right=of li, xshift=-2em] (sj) {$\cmspta_j^1$};

		\node[invariantNorth] at (si.north) {$\clockCMcone \leq \paramp - 1$};
		\node[invariantNorth] at (li.north) {$\clockCMcone \leq \paramp$};
		\node[invariantNorth] at (sj.north) {$\cdots$}; %

		\path[edge] (si) edge[] node[above] {$\clockCMcone = \paramp - 1$ } node[below]{ $ \clockCMcone \assign 0$} (li);
		\path[edge] (li) edge[] node[above,locguard] {$\locguardof{\cmspta_j^{\clockCMt}}$ } node[below]{$\clockCMcone \leq \paramp$} (sj);
	\end{tikzpicture}
	}

	\caption{Increment~$\Counter_1$: gadget simulating~$\clockCMcone$}
	\label{figure:2CM:increment:x1}
	\end{subfigure}
	\begin{subfigure}[b]{.5\textwidth}
	\centering
	\scalebox{\generalFigsScaleFactor}{
	\begin{tikzpicture}[PTA, node distance=1.8cm]
		\node[location2CM] (si) {$\cmspta_i^2$};
		\node[location, right=of si] (li) {$\locij{i}{2}$};
		\node[location2CM, right=of li] (sj) {$\cmspta_j^2$};

		\node[invariantNorth] at (si.north) {$\clockCMctwo \leq \paramp$};
		\node[invariantNorth] at (li.north) {$\clockCMctwo \leq \paramp$};
		\node[invariantNorth] at (sj.north) {$\cdots$}; %

		\path[edge] (si) edge[] node[above] {$\clockCMctwo = \paramp$ } node[below]{ $ \clockCMctwo \assign 0$} (li);
		\path[edge] (li) edge[] node[above,locguard] {$\locguardof{\cmspta_j^{\clockCMt}}$ } node[below]{$\clockCMctwo \leq \paramp$} (sj);
	\end{tikzpicture}
	}

	\caption{Increment~$\Counter_1$: gadget simulating~$\clockCMctwo$}
	\label{figure:2CM:increment:x2}
	\end{subfigure}

	\caption{Increment gadget}
	\label{figure:2CM:increment}
\end{figure*}
\begin{figure*}[tb]
	\begin{subfigure}[b]{.5\textwidth}
	\centering
	\scalebox{\generalFigsScaleFactor}{
		\begin{tikzpicture}[PTA, node distance=1.3cm]
			\node[location2CM] (si) {$\cmspta_i$};
			\node[location, xshift=0, right=of si] (l1) {$\loci{i1}$};
			\node[location, above right=of l1] (l2) {$\loci{i2}$};
			\node[location, below right=of l1] (l2p) {$\loci{i2}'$};
			\node[location, below right=of l2, xshift=-.5em] (l3) {$\loci{i3}$};
			\node[location2CM, right=of l3] (sj) {$\cmspta_j$};
			\node[location, above=of si,yshift=-2em] (lk1) {$\loci{k1}$};
			\node[location, above=of lk1,yshift=-1.5em] (lk2) {$\loci{k2}$};
			\node[location2CM, right=of lk2] (sk) {$\cmspta_k$};

			\path[edge] (si) edge node {$\clockCMt = 0$} node[below]{$\land \clockCMcone > 0$} (l1);
			\path[edge] (l1) edge node[align=center] {$\clockCMctwo = \paramp$ \\ $ \clockCMctwo \assign 0$} (l2);
			\path[edge] (l1) edge node[swap,align=center] {$\clockCMcone = \paramp + 1$ \\ $ \clockCMcone \assign 0$} (l2p);
			\path[edge] (l2) edge node[align=center] {$\clockCMcone = \paramp + 1$ \\ $ \clockCMcone \assign 0$} (l3);
			\path[edge] (l2p) edge node[swap,align=center] {$\clockCMctwo = \paramp$ \\ $ \clockCMctwo \assign 0$} (l3);
			\path[edge] (l3) edge node[align=center, above] {$\clockCMt = \paramp$ } node[below]{$ \clockCMt\assign 0$} (sj);
			\path[edge] (si) edge node[align=center] {$\clockCMt = 0$ \\ $\land \clockCMcone = 0$} (lk1);
			\path[edge] (lk1) edge node[align=center] {$\clockCMctwo = \paramp$\\$\clockCMctwo \assign 0$} (lk2);
			\path[edge] (lk2) edge node[above] {$\clockCMt = \paramp$} node[below]{$\clockCMt, \clockCMcone \assign 0$} (sk);
		\end{tikzpicture}
	}

	\caption{0-test and decrement ($\Counter_1$): single PTA version}
	\label{figure:2CM:original:decrement}
	\end{subfigure}
	\hfill{}
	\begin{subfigure}[b]{.66\textwidth}
	\centering
	\scalebox{\generalFigsScaleFactor}{
	\begin{tikzpicture}[PTA, node distance=1.5cm]
		\node[location2CM] (si) {$\cmspta_i^{\clockCMt}$};
		\node[location, right=of si, xshift=1em] (li) {$\locij{i}{\clockCMt}$};
		\node[location2CM, right=of li] (sj) {$\cmspta_j^{\clockCMt}$};
		\node[location, above=of si, yshift=-2.5em] (lk) {$\locij{k}{\clockCMt}$};
		\node[location2CM, above=of lk, yshift=-2em] (sk) {$\cmspta_k^{\clockCMt}$};

		\node[invariantSouth] at (si.south) {$\clockCMt = 0$};
		\node[invariantNorth] at (li.north) {$\clockCMt \leq \paramp$};
		\node[invariantNorth] at (sj.north) {$\clockCMt = 0$};
		\node[invariantWest] at (lk.west) {$\clockCMt \leq \paramp$};
		\node[invariantNorth] at (sk.north) {$\clockCMt = 0$};

		\path[edge] (si) edge[] node[above,locguard]{$\locguardof{\locij{i1}{1}}$} node[below] {$\clockCMt = 0$}  (li);
		\path[edge] (li) edge[] node[above] {$\clockCMt = \paramp$ } node[below]{ $ \clockCMt \assign 0$} (sj);
		\path[edge] (si) edge[] node[align=center] {$\clockCMt = 0$} node[right,locguard,xshift=.2em]{$\locguardof{\locij{k}{1}}$}  (lk);
		\path[edge] (lk) edge[] node[align=center] {$\clockCMt = \paramp$\\$\clockCMt \assign 0$} (sk);
	\end{tikzpicture}
	}

	\caption{0-test and decrement ($\Counter_1$): gadget simulating~$\clockCMt$}
	\label{figure:2CM:decrement:t}
	\end{subfigure}
	\begin{subfigure}[b]{.5\textwidth}
	\centering
	\scalebox{\generalFigsScaleFactor}{
	\begin{tikzpicture}[PTA, node distance=1.6cm]
		\node[location2CM] (si) {$\cmspta_i^1$};
		\node[location, right=of si] (li1) {$\locij{i1}{1}$};
		\node[location, right=of li1] (li2) {$\locij{i2}{1}$};
		\node[location2CM, right=of li2] (sj) {$\cmspta_j^1$};
		\node[location, above=of si, yshift=-2.8em] (lk) {$\locij{k}{1}$};
		\node[location2CM, above=of lk, yshift=-2.2em] (sk) {$\cmspta_k^{1}$};

		\node[invariantSouth] at (si.south) {$\clockCMcone \leq \paramp + 1$};
		\node[invariantNorth] at (li1.north) {$\clockCMcone \leq \paramp + 1$};
		\node[invariantNorth] at (li2.north) {$\clockCMcone \leq \paramp$};
		\node[invariantNorth] at (sj.north) {$\cdots$}; %
		\node[invariantWest] at (lk.west) {$\clockCMcone \leq \paramp$};
		\node[invariantWest] at (sk.west) {$\cdots$}; %

		\path[edge] (si) edge[] node[above,locguard]{$\locguardof{\cmspta_i^{\clockCMt}}$ } node[below] {$\clockCMcone > 0$ } (li1);
		\path[edge] (li1) edge[] node[above] {$\clockCMcone = \paramp + 1$ } node[below]{ $ \clockCMcone \assign 0$} (li2);
		\path[edge] (li2) edge[] node[above,locguard] {$\locguardof{\cmspta_j^{\clockCMt}}$ } node[below]{$\clockCMcone \leq \paramp$} (sj);
		\path[edge] (si) edge[] node[left,align=center] {$\clockCMcone = 0$} node[right,locguard,xshift=.2em]{$\locguardof{\cmspta_i^{\clockCMt}}$ } (lk); %
		\path[edge] (lk) edge[] node[left,align=center] {$\clockCMcone = \paramp$\\$\clockCMcone \assign 0$} node[right,locguard,xshift=.2em]{$\locguardof{\cmspta_k^{\clockCMt}}$ } (sk);
	\end{tikzpicture}
	}

	\caption{0-test and decrement ($\Counter_1$): gadget simulating~$\clockCMcone$}
	\label{figure:2CM:decrement:x1}
	\end{subfigure}
	\begin{subfigure}[b]{.66\textwidth}
	\centering
	\scalebox{\generalFigsScaleFactor}{
	\begin{tikzpicture}[PTA, node distance=1.5cm]
		\node[location2CM] (si) {$\cmspta_i^2$};
		\node[location, right=of si, xshift=2em] (li) {$\locij{i}{2}$};
		\node[location2CM, right=of li] (sj) {$\cmspta_j^2$};
		\node[location, above=of si, yshift=-2em] (lk) {$\locij{k}{2}$};
		\node[location2CM, above=of lk, yshift=-2.5em] (sk) {$\cmspta_k^{2}$};

		\node[invariantSouth] at (si.south) {$\clockCMctwo \leq \paramp$};
		\node[invariantNorth] at (li.north) {$\clockCMctwo \leq \paramp$};
		\node[invariantNorth] at (sj.north) {$\cdots$}; %
		\node[invariantWest] at (lk.west) {$\clockCMctwo \leq \paramp$};
		\node[invariantWest] at (sk.west) {$\cdots$}; %

		\path[edge] (si) edge[] node[above,locguard,yshift=1em]{$\locguardof{\locij{i}{\clockCMt}}$} node[above] {$\clockCMctwo = \paramp$ } node[below]{ $ \clockCMctwo \assign 0$} (li);
		\path[edge] (li) edge[] node[above,locguard] {$\locguardof{\cmspta_j^{\clockCMt}}$ } node[below]{$\clockCMctwo \leq \paramp$} (sj);
		\path[edge] (si) edge[] node[left,align=center] {$\clockCMctwo = \paramp$\\$\clockCMctwo \assign 0$} node[right,locguard,xshift=.2em]{$\locguardof{\locij{k}{\clockCMt}}$} (lk);
		\path[edge] (lk) edge[] node[left] {$\clockCMctwo \leq \paramp$} node[right,locguard,xshift=.2em]{$\locguardof{\cmspta_k^{\clockCMt}}$ } (sk);
	\end{tikzpicture}
	}

	\caption{0-test and decrement ($\Counter_1$): gadget simulating~$\clockCMctwo$}
	\label{figure:2CM:decrement:x2}
	\end{subfigure}

	\caption{0-test and decrement gadget}
	\label{figure:2CM:decrement}
\end{figure*}

	Let us now consider the 0-test and decrement gadget.
	Decrement is done similarly to increment, by replacing guards $\clockCMcone = \param - 1$ with $\clockCMcone = \param + 1$, as shown in \cref{figure:2CM:original:decrement}.
    In addition, the 0-test is obtained by simply testing that $\clockCMcone = 0$ whenever $\clockCMt = 0$ (which ensures that~$\countervalue_1 = 0$), which is done on the guard from~$\cmspta_j$ to~$\loci{k1}$; we then force exactly $\param$ time units to elapse (and reset each clock when it reaches~$\param$), which means that the values of the clocks when leaving the gadget are identical to their value when entering.
    This is not strictly speaking needed here, but this time elapsing will simplify the proof of \cref{proposition:undecidability:3-processes-invariants}.
    Dually, the guard from $\cmspta_i$ to~$\loci{i1}$ ensures that decrement is done only when the counter is not null.

	All those gadgets also work for counter~$\Counter_2$ by swapping $\clockCMcone$ and~$\clockCMctwo$.

	The actions associated with the edges do not matter; we can assume a single action~$\action$ on all edges (omitted in all figures).

	We now prove that the machine halts iff there exists a parameter valuation~$\pval$ such that $\valuate{\PTA}{\pval}$ reaches location~$\lochalt$.
	First note that if $\param \leq 1$ the initial gadget cannot be passed, and so the machine does not halt.
	Assume $\param > 1$.
	Consider two cases:
	\begin{enumerate}
		\item either the value of the counters is not bounded (and the 2CM does not halt).
			Then, for any parameter valuation, at some point during an increment of, say, $\Counter_1$ we will have $\countervalue_1 > \param$ and hence $\clockCMcone > \param - 1$ when taking the edge from $\loci{i2}$ to $\loci{i3}$ and the PTA will be blocked.
			Therefore, there exists no parameter valuation for which the PTA can reach~$\lochalt$.
		\item or the value of the counters remains bounded.
			Let $\maxcounterval$ be their maximal value.
			Let us consider two subcases:
			\begin{enumerate}
				\item either the machine reaches $\cmshalt$: in that case, if $\maxcounterval \leq \param$, then the PTA valuated with such parameter valuations correctly simulates the machine, yielding a (unique) run reaching location~$\lochalt$.

				\item or the machine does not halt.
                    Then again, for a sufficiently large parameter valuation (\ie{} $\param > \maxcounterval$), the machine is properly simulated, and since the machine does not halt, then the PTA never reaches~$\lochalt$.
					For other values of~$\param$, the machine will block at some point in an increment gadget, because $\param$ is not large enough and the guard in the increment gadget cannot be satisfied.
			\end{enumerate}
	\end{enumerate}
	Hence the machine halts iff there exists a parameter valuation~$\pval$ such that $\valuate{\PTA}{\pval}$ reaches~$\lochalt$. 
\end{proof}
\begin{remark}
	Observe that the proof of \cref{lemma:undecidability:reach-emptiness:integers} does not require invariants, so undecidability holds without invariants as well---a result known since~\cite{AHV93}.
\end{remark}

We now prove undecidability of the reachability-emptiness in NTA with exactly 3~processes, by rewriting the 2CM encoding from the former proof.

\begin{restatable}{proposition}{propositionUndecThreeInv}\label{proposition:undecidability:3-processes-invariants}
	Assume a \gPTA{} $\PTA$ with invariants, with a single clock and a single parameter, and $\loc$ one of its locations.
	Deciding the emptiness of the set of timing parameter valuations~$\pval$ for which $\loc$ is reachable in the NTA $\network{\valuate{\PTA}{\pval}}{3}$ is undecidable, for all such~$\PTA$.
\end{restatable}
\begin{proof}
	The proof main technicality is to rewrite the 2CM encoding of \cref{lemma:undecidability:reach-emptiness:integers} (made of a PTA with 3~clocks) using 3 processes with 1~clock each.
	This is not trivial as the communication model between our \gPTAs{} is rather weak.
	Recall that the parameter encodes typically the maximal valuation of the counters, and can therefore be arbitrarily large; we can assume without loss of generality that $\param > 1$.
	Here, we use three different \gPTAs{} (each featuring a single clock) synchronizing together; of course, strictly speaking, we can use only a single structure, but with three ``subparts'', and we ensure that each of the 3~processes will go into a different subpart of the \gPTA{}.
	This is ensured by the initial gadget in \cref{figure:2CM:initial:full}: in order for process~1 to reach location~$\cmspta_0^{\clockCMt}$, then, due to the location guards, exactly one other process must have selected the second subpart (location~$\locij{1}{1}$) while the third process must have selected the third subpart (location~$\locij{1}{2}$).

	Each instruction of the 2CM is encoded into a gadget connected with each other, for each of the three subparts.

	In the rest of the proof, we therefore give, for each instruction of the 2CM, one gadget for each of the three subparts.
	Obviously, the (unique) \gPTA{} features a single clock (``$\clockx$'' in \cref{figure:2CM:initial:full});  however, for sake of readability, and for consistency with the encoding given in \cref{lemma:undecidability:reach-emptiness:integers}, we use clocks with different names in the 3 subparts: that is, the initial gadget in \cref{figure:2CM:initial:full} can be described by three gadgets for the three subparts, given in \cref{figure:2CM:initial:t,figure:2CM:initial:x1,figure:2CM:initial:x2}, with clock names $\clockCMt$, $\clockCMcone$ and~$\clockCMctwo$.
	Also recall that from \cref{definition:gPTA}, clocks are not shared---they cannot be read from another process.

	Let us consider the increment gadget: the decomposition into three ``subparts'' (for each of the three processes) is given in \cref{figure:2CM:increment:t,figure:2CM:increment:x1,figure:2CM:increment:x2}.
	An invariant ``$\cdots$'' denotes the fact that the actual invariant is given in the ``next'' gadget starting from that location.
	Note that it takes exactly $\param$ time units to move from~$\cmspta_i^{\clockCMt}$ to~$\cmspta_j^{\clockCMt}$.
	In addition, the outgoing transition from a given~$\cmspta_i^{\clockCMt}$
	is always done in 0-time, which is enforced by the invariant $\clockCMt = 0$;
	therefore, the location guard $\locguardof{\cmspta_j^{\clockCMt}}$ together with the guard $\clockCMctwo \leq \param$ is exactly equivalent to $\clockCMt = \param \land \clockCMctwo \leq \param$, forcing the subparts simulating~$\clockCMcone$ and $\clockCMctwo$ to properly synchronize with the subpart simulating~$\clockCMt$ (or getting stuck forever in $\locij{i}{1}$ or~$\locij{i}{2}$ and blocking the whole computation).
	Therefore, %
	all three gadgets synchronize as expected, simulating the original increment gadget in \cref{figure:2CM:original:increment}.

	Let us now consider the 0-test and decrement gadget: the decomposition into three ``subparts'' (for each of the three processes) is given in \cref{figure:2CM:decrement:t,figure:2CM:decrement:x1,figure:2CM:decrement:x2}.
	This time, compared to the former gadgets, we need a better synchronization between gadgets to ensure the correct branching depending on the value of~$\countervalue_1$.
	If $\countervalue_1 = 0$, then $\clockCMcone = 0$ when $\clockCMt = 0$, and therefore the process simulating~$\clockCMcone$ can move in 0-time to~$\locij{k}{1}$, since~$\cmspta_i^{\clockCMt}$ is occupied;
	then, again in 0-time, the process simulating~$\clockCMt$ can move to~$\locij{k}{\clockCMt}$, ensuring the correct synchronization between both processes in the correct branch.
	Note that, because of the invariant in~$\cmspta_i^{\clockCMt}$, the process simulating~$\clockCMt$ can only stay 0 time unit in~$\cmspta_i^{\clockCMt}$, and therefore in the process simulating~$\clockCMcone$ the guard $\clockCMcone = 0$ together with location guard~$\locguardof{\cmspta_j^{\clockCMt}}$ is exactly equivalent to~$\clockCMcone = 0 \land \clockCMt = 0$, which correctly encodes the 0-test.
	Conversely, if $\countervalue_1 > 0$, then the process simulating~$\clockCMcone$ can move in 0-time to~$\locij{i1}{1}$, since~$\cmspta_i^{\clockCMt}$ is occupied;
	then, again in 0-time, the process simulating~$\clockCMt$ can move to~$\locij{i}{\clockCMt}$.
	Similarly, the process simulating~$\clockCMctwo$ follows the right branch thanks to location guards $\locguardof{\locij{i}{\clockCMt}}$ and~$\locguardof{\locij{k}{\clockCMt}}$.
	The rest of the decrement part (from~$i$ to~$j$) works similarly to the increment gadget, and the correctness argument is the same.
    Recall that the rest of the 0-test gadget forces time to elapse for $\param$~units, without modifying the value of the three clocks when leaving the gadget (except of course for $\clockCMcone$ in case of decrement).
	In particular, in the gadget simulating~$\clockCMctwo$ (\cref{figure:2CM:decrement:x2}), $\clockCMctwo$~has exactly the same value when entering~$\cmspta_k^{2}$ as when entering $\cmspta_i^{2}$ due to location guard~$\locguardof{\cmspta_k^{\clockCMt}}$, and the intermediate location resetting~$\clockCMctwo$ when it is exactly~$\param$;
	the same holds in the decrement part of the gadget.

	In the encoding in the proof of \cref{lemma:undecidability:reach-emptiness:integers}, we showed that the 2CM reaches~$\cmshalt$ iff the location~$\lochalt$ is reachable.
	Here, due to the split of the encoding into three subparts, we need a final gadget, given in \cref{figure:2CM:final} (we assume the clock is reset when entering locations $\cmspta_{\cmshaltname}^{\clockCMt}$, $\cmspta_{\cmshaltname}^{1}$ and~$\cmspta_{\cmshaltname}^{2}$).
	In order for the first process to reach the (new) location~$\lochalt$, we need all three processes to reach their respective final location (\ie{} $\cmspta_{\cmshaltname}^{\clockCMt}$, $\cmspta_{\cmshaltname}^{1}$ and~$\cmspta_{\cmshaltname}^{2}$) together, which is easily achieved thanks to the various location guards done in 0-time.
	(Note that this gadget actually allows \emph{all} processes to reach~$\cmshalt$; so far, we only need one process to do so.)

	Now, first note that the subpart simulating~$\clockCMt$ progresses in its gadgets, only synchronizing with the two other processes in the initial and final gadgets, as well as in the 0-test and decrement.
	In contrast, the two subparts simulating~$\clockCMcone$ and~$\clockCMctwo$ constantly synchronize with the location guards from the first subpart; while nothing forces them to take these transitions (notably those guarded by $\locguardof{\cmspta_j^{\clockCMt}}$), failing in taking such a transition will immediately lead to a timelock due to the invariants and to the fact that the next time such a location guard will be available is necessarily in $\param$ time units.
	More in detail, recall that any guard location (location used in a location guard)~$\cmspta_i^{\clockCMt}$ in the subpart simulating~$\clockCMt$ has an invariant $\clockCMt = 0$, and only outgoing transitions guarded with $\clockCMt = 0$.
	In addition, at least $\param$ time units must elapse between two guard locations in the subpart simulating~$\clockCMt$: so, due to the invariants in the two other subparts, failing to take a location guard renders impossible to take it the next time the first subpart will be in a guard location.

	For all these reasons, if one process reaches~$\lochalt$, then from \cref{figure:2CM:final}, two other processes must have traversed the two other subparts (simulating~$\clockCMcone$ and~$\clockCMctwo$) correctly, by synchronizing with the first process.
	Hence, the 2CM is correctly simulated.

	The rest of the reasoning follows the proof of \cref{lemma:undecidability:reach-emptiness:integers}:
	the 2CM halts iff there exists a parameter valuation~$\pval$ such that $\network{\valuate{\PTA}{\pval}}{3}$ reaches~$\lochalt$.
\end{proof}

We show that the absence of invariants does not avoid undecidability, for exactly 3~processes.

\begin{restatable}{proposition}{propositionUndecThreeNoInv}\label{proposition:undecidability:3-processes-noinvariants}
	Assume a \gPTA{} $\PTA$ with a single clock and a single parameter, and $\loc$ one of its locations.
	Deciding the emptiness of the set of timing parameter valuations~$\pval$ for which $\loc$ is reachable in the NTA $\network{\valuate{\PTA}{\pval}}{3}$ is undecidable, even without invariants, for all such~$\PTA$.
\end{restatable}
\begin{proof}[Proof sketch]
	The idea is that, if a process ``chooses'' to not take some outgoing transition in the absence of invariants, then the whole computation is stuck, and $\lochalt$ is unreachable.
	Put it differently, while the absence of invariants may add some runs, none of these runs can reach~$\lochalt$, and the final correctness argument of the proof of \cref{proposition:undecidability:3-processes-invariants} still holds.
	See \cref{appendix:proof:proposition:undecidability:3-processes-noinvariants}.
\end{proof}
\subsubsection{Undecidability for any fixed number of processes (with or without invariants)}\label{section:undecidability:fixed}

By modifying the constructions in the proof of \cref{proposition:undecidability:3-processes-noinvariants}, it is easy to show that, \emph{for any fixed number of processes $n \geq 3$}, the reachability-emptiness problem is undecidable, even without invariants.
Note that our construction differs for any $n \geq 3$, and hence this does not prove the undecidability of the parametrised-reachability-emptiness problem (which is actually \emph{decidable} without invariants, see \cref{theorem:decidability:no-invariant}).

\begin{restatable}{proposition}{propositionTwoCMUndecFixed}\label{proposition:2CM:fixedn}
	Let $\PTA$ be a \gPTA{}, and $\loc$ one of its locations.
	For any~$n \geq 3$, deciding the emptiness of the set of timing parameter valuations~$\pval$ for which $\loc$ is reachable in the NTA~$\network{\valuate{\PTA}{\pval}}{n}$ is undecidable, even when $\PTA$ contains a single clock and a single parameter without invariants.
\end{restatable}
\begin{proof}[Proof sketch]
	The idea is to reuse the construction of the proof of \cref{proposition:undecidability:3-processes-noinvariants}, by ``killing'' any process except~3 by sending them to ``sink'' locations.
	See \cref{appendix:proof:proposition:2CM:fixedn}.
\end{proof}
\subsection{Parameterized reachability-emptiness with invariants}\label{ss:undecidability-local-invariants}
Let us now consider the parametrised reachability-emptiness problem in \PDTNs{}.
In the absence of invariants, the 2CM encoding used in the previous proofs would become incorrect for a parametric number of processes: recall that we avoided the use of invariants by letting the system get stuck if a transition is not taken at the ``right'' time.
However, in a system with a parametric number of processes, a single process that gets stuck will not cause the whole computation to get stuck, since another process may have taken the transition at the ``right'' time.
Moreover, a process that gets stuck will ruin the synchronization that guarantees simulation of the~2CM, since now taking a transition with location guard~$\locguardof{\cmspta_i^{\clockCMt}}$ will be possible at any time after the process gets stuck in~$\cmspta_i^{\clockCMt}$.
We now prove that undecidability holds however \emph{with} invariants.

\begin{restatable}{theorem}{theoremUndecInv}\label{theorem:undecidability-local-invariants}
	PR-emptiness is undecidable for general \PDTNs{}, even with a single clock and a single parameter (with invariants).
\end{restatable}
\begin{proof}[Proof sketch]
	We rely on the construction given in the proof of \cref{proposition:undecidability:3-processes-invariants}, and we show that any additional processes will just simulate one of the existing three processes, due to the presence of invariants.
	Most importantly, no process can remain ``idle'' in a location forever, and therefore the location guards still guarantee that the 2CM is simulated correctly.
	See \cref{appendix:proof:theorem:undecidability-local-invariants}.
\end{proof}
\subsection{Parameterized global reachability-emptiness with or without invariants}\label{ss:undecidability-global-noinvariants}

We show here that our former constructions can get rid of invariants providing we consider \emph{global} properties.
That is, without invariants, global properties make parametrised reachability-emptiness undecidable while it is decidable for local properties (see \cref{section:decidability:noinvariants}).

\begin{restatable}{theorem}{theoremUndecGlobal}\label{theorem:undecidability-global-noinvariants}
	PGR-emptiness is undecidable for general \PDTNs{}, even with a single clock and a single parameter, with or without invariants.
\end{restatable}
\begin{proof}[Proof sketch]
	The proof technical idea is that we ask whether \emph{all} processes can reach the target location~$\lochalt$, which can only be done if the 2CM was properly simulated.
	See \cref{appendix:proof:theorem:undecidability-global-noinvariants}.
\end{proof}
\section{Decidability results}\label{section:decidability}

We will first show that,
for two subclasses of \PDTNs{}, \ie{} fully parametric \PDTNs{} with a single parameter (\cref{ss:fpPDTN})
and
L/U-\PDTNs{} (\cref{ss:LUPDTNs}),
deciding PR-emptiness (resp.\ PGR-emptiness)
is equivalent to checking parametrised reachability (resp.\ PGR) in a non-parametric \DTN{}---which is decidable.
The third positive result, addressing one clock, arbitrarily many parameters and no invariants (\cref{section:decidability:noinvariants}), is more involved and requires different techniques.

\subsection{Fully parametric \PDTNs{} with a single parameter}\label{ss:fpPDTN}
\begin{restatable}{theorem}{theoremDecidabilityfpPTA}\label{theorem:decidability:fpPTA}
	PR-emptiness (resp.\ PGR-emptiness) for fully parametric \PDTNs{} with 1~parameter and arbitrarily many clocks is equivalent to PR (resp.\ PGR) for \DTNs{}.
\end{restatable}
\begin{proof}[Proof sketch]
	The idea is to test the satisfaction of the property on only two non-parametric \DTNs{}, \viz{} these valuating the only parameter with 0 and~1.
	We show that any other non-zero valuation is equivalent (up to rescaling) to the valuation~1:
		indeed, without constant terms in the model (other than~0), multiplying the value of the (unique) parameter by~$n$ will only impact the value of the duration of the runs by~$n$, but will not impact reachability.
	See \cref{appendix:proof:theorem:decidability:fpPTA}.
\end{proof}
\subsection{L/U-\PDTNs{} with arbitrarily many clocks and parameters}\label{ss:LUPDTNs}

Given
$d \in \setZ$,
let $\pvalzerod{d}$ denote the valuation such that
	for all $\param \in \ParamSetL$, $\pval(\param) = 0$ and
	for all $\param \in \ParamSetU$, $\pval(\param) = d$.
We will consider the special valuation $\pvalzeroinf$: strictly speaking, $\pvalzeroinf$ is not a proper parameter valuation due to~$\infty$, but valuating a PTA with $\pvalzeroinf$ consists in removing the inequalities involving upper-bound parameters with a positive coefficient.

\begin{restatable}{theorem}{theoremDecidabilityLU}\label{theorem:L/U}
	PR-emptiness (resp.\ PGR-emptiness) for L/U-\PDTNs{}
	is equivalent to PR (resp.\ PGR) for \DTNs{}.
\end{restatable}
\begin{proof}[Proof sketch]
	We first show that networks of L/U-\gPTAs{} are \emph{monotonic}:
	that is,
	given $\PTA$ an L/U-\gPTA{} and $\pval$ a parameter valuation,
	any computation of $\network{(\valuate{\PTA}{\pval})}{n}$ is a computation of $\network{(\valuate{\PTA}{\pval'})}{n}$,
	for all $\pval'$ such that
		upper-bound parameters are larger than or equal to their value in~$\pval$
		and
		lower-bound parameters are smaller than or equal to their value in~$\pval$.
	We then show that, given $\varphi$ a global reachability property over~$\PTA$, PGR-emptiness does not hold iff $\varphi$ is satisfied in a configuration reachable in the DTN~$\network{(\valuate{\PTA}{\pvalzeroinf})}{\infty}$, as this DTN contains the behaviours for \emph{all} parameter valuations, due to the monotonicity.
	Therefore, deciding PGR-emptiness for L/U-\PDTNs{} amounts to deciding satisfaction of~$\varphi$ in the DTN $\network{(\valuate{\PTA}{\pvalzeroinf})}{\infty}$; an additional technicality is necessary to show that if $\varphi$ is satisfied in a configuration reachable in $\network{(\valuate{\PTA}{\pvalzeroinf})}{\infty}$, then it is also reachable in $\network{(\valuate{\PTA}{\pval})}{\infty}$ for some concrete (non-$\infty$) parameter valuation~$\pval$.
	See \cref{appendix:proof:theorem:L/U}.
\end{proof}

Since PR and PGR can be solved in \EXPSPACE{} and \twoEXPSPACE{} respectively in \DTNs{}~\cite{AJKS25},
then solving PR-emptiness and PGR-emptiness can be done with the same complexities for fully parametric \PDTNs{} with a single parameter and for L/U-\PDTNs{}.

\subsection{\PDTNs{} with one clock, arbitrarily many parameters and no invariants}\label{section:decidability:noinvariants}

Our last decidability result closes the gap of \cref{section:undecidability}:
while PR-emptiness for \PDTNs{} even with a single clock and one parameter
is undecidable for local properties with invariants,
and
is undecidable for global properties even without invariants,
we show that the problem becomes decidable for local properties without invariants.
This result highlights the power of invariants in \PDTNs{}.

\begin{restatable}{theorem}{theoremDecidabilityNoInvariant}\label{theorem:decidability:no-invariant}
	PR-emptiness is decidable for \PDTNs{} with a single clock, arbitrarily many parameters, and no invariants.
\end{restatable}
\begin{proof}[Proof sketch]
	We use here a completely different construction:
	we write a formula in the existential fragment of the first order theory of the integers with addition (a.k.a.\ Presburger arithmetic) enhanced with the divisibility operand (a decidable logic~\cite{LOW15}), which will include the computation of the reachable durations of the locations used in guards in parallel to the reachability of the final location.
	In order to write this formula, we rely on results relating affine parametric semi-linear sets (apSl sets), a parametric extension of semi-linear sets, to durations in a PTA.
	Solving PR-emptiness then boils down to deciding the truth of the formula.
		
	This proof technique depends crucially on the parameters being integers. The formula we construct intermixes these parameters with integer variables that encode looping behaviours in the \PDTN{}. If the parameters were allowed to take rational values, the resulting formula would fall outside the logical theories currently known to be decidable.

	See \cref{appendix:proof:theorem:decidability:no-invariant} for our full proof.
\end{proof}
\section{Conclusion}\label{section:conclusion}
\begin{table}[tb]
	\centering
	\caption{Decidability of PR-e and PGR-e for \PDTNs{} (\cellDecidable{} $=$ decidability, \cellUndecidable{} $=$ undecidability)}
	\label{table-summary}
	\setlength{\tabcolsep}{2pt} %
	\begin{tabular}{| c | c | c | c | c | c |}
		\hline
		\rowHeader{}\cellHeader{Clocks} & \cellHeader{Parameters} & \cellHeader{Invariants} & \cellHeader{Local properties} & \cellHeader{Global properties} \\
		\hline
		\multirow{2}{*}{1} & \multirow{2}{*}{arbitrary} & \cellHeader{Without} & \cellDecidable{} (\cref{theorem:decidability:no-invariant}) & \cellUndecidable{} (\cref{theorem:undecidability-global-noinvariants})\\
		\cline{3-5}
		& & \cellHeader{With} & \cellUndecidable{} (\cref{theorem:undecidability-local-invariants}) &  \cellUndecidable{}\\
		\hline
		\multirow{2}{*}{arbitrary} & \multirow{2}{*}{1-param.\ fully-parametric} & \cellHeader{Without} & \cellDecidable{} (\cref{theorem:decidability:fpPTA} + \cite{AEJK24}) & \cellOpen{} (\cref{theorem:decidability:fpPTA})\\
		\cline{3-5}
		& & \cellHeader{With} & \cellDecidable{} (\cref{theorem:decidability:fpPTA} + \cite{AJKS25}) & \cellOpen{} (\cref{theorem:decidability:fpPTA})\\
		\hline
		\multirow{2}{*}{arbitrary} & \multirow{2}{*}{L/U} & \cellHeader{Without} & \cellDecidable{} (\cref{theorem:decidability:fpPTA} + \cite{AEJK24}) & \cellOpen{} (\cref{theorem:decidability:fpPTA})\\
		\cline{3-5}
		& & \cellHeader{With} & \cellDecidable{} (\cref{theorem:decidability:fpPTA} + \cite{AJKS25}) & \cellOpen{} (\cref{theorem:decidability:fpPTA})\\
		\hline
	\end{tabular}
\end{table}
In this paper, we investigated models with uncertainty over timing parameters, and parametrised in their number of components.
We showed that the emptiness of the parameter valuations set for which a given location is reachable for some number of processes is decidable for networks of \gPTAs{} with a single clock in each process and arbitrarily many parameters, provided no invariants are used, and for local properties only; this positive result is tight, in the sense that adding invariants or considering global properties makes this problem undecidable.
We additionally exhibited two further decidable subclasses, by restricting the use of the parameters, without restriction of the number of clocks and parameters.
We summarize our results in \cref{table-summary}.
Beside emphasizing the strong power of invariants and global properties, our results show an interesting fact on the expressive power of the communication model: while reachability is decidable for PTAs with 1 or 2 clocks and one parameter~\cite{AHV93,ALM20,GH24}, it becomes undecidable in our setting with a single parameter.

\fakeParagraph{Perspectives}
First, it can be easily shown that our negative results from \cref{section:undecidability} all hold as well over discrete time.
Second, our undecidability results are expressed so far over integer-valued parameters.
While the undecidability of integer-valued parameters implies undecidability for rational-valued parameters, the case of \emph{bounded} rational-valued parameters (\eg{} in a closed interval) remains open; we conjecture it remains undecidable using a different encoding of our undecidability proofs.
However, we have no hint regarding the extension of our decidable results (notably \cref{theorem:decidability:no-invariant}), as our proof techniques heavily rely on the parameter integerness.
Also, whether extending the decidable case of \cref{theorem:decidability:no-invariant} to two clocks preserves decidability remains open too.

An interesting future work is the question of \emph{parameter synthesis} (for the decidable cases).
That is, beyond deciding the emptiness of the valuations set for which some properties hold, can we \emph{synthesize} them?

Finally, considering \emph{universal} properties is an interesting perspective: this can include the existence of a parameter valuation for which \emph{all} numbers of processes satisfy a property or, conversely, the fact that \emph{all} parameter valuations are such that some number of processes (or all processes) satisfy a property.

\newpage

	\newcommand{\CCIS}{Communications in Computer and Information Science}
	\newcommand{\ENTCS}{Electronic Notes in Theoretical Computer Science}
	\newcommand{\FAC}{Formal Aspects of Computing}
	\newcommand{\FundInf}{Fundamenta Informaticae}
	\newcommand{\FMSD}{Formal Methods in System Design}
	\newcommand{\IJFCS}{International Journal of Foundations of Computer Science}
	\newcommand{\IJSSE}{International Journal of Secure Software Engineering}
	\newcommand{\IPL}{Information Processing Letters}
	\newcommand{\JAIR}{Journal of Artificial Intelligence Research}
	\newcommand{\JLAP}{Journal of Logic and Algebraic Programming}
	\newcommand{\JLAMP}{Journal of Logical and Algebraic Methods in Programming} %
	\newcommand{\JLC}{Journal of Logic and Computation}
	\newcommand{\LMCS}{Logical Methods in Computer Science}
	\newcommand{\LNCS}{Lecture Notes in Computer Science}
	\newcommand{\RESS}{Reliability Engineering \& System Safety}
	\newcommand{\RTS}{Real-Time Systems}
	\newcommand{\SCP}{Science of Computer Programming}
	\newcommand{\SOSYM}{Software and Systems Modeling ({SoSyM})}
	\newcommand{\STTT}{International Journal on Software Tools for Technology Transfer}
	\newcommand{\TCS}{Theoretical Computer Science}
	\newcommand{\TOPLAS}{{ACM} Transactions on Programming Languages and Systems ({ToPLAS})}
	\newcommand{\ToPNoC}{Transactions on {P}etri Nets and Other Models of Concurrency}
	\newcommand{\TOSEM}{{ACM} Transactions on Software Engineering and Methodology ({ToSEM})}
	\newcommand{\TSE}{{IEEE} Transactions on Software Engineering}
	\bibliographystyle{plainurl}%
	\bibliography{PDTN}
\newpage
\appendix
\section{Proofs of \cref{section:undecidability}}
\subsection{Proof of \cref{proposition:undecidability:3-processes-noinvariants}\label{appendix:proof:proposition:undecidability:3-processes-noinvariants}}
\propositionUndecThreeNoInv*
\begin{proof}
	Let us show that the absence of invariants in the proof of \cref{proposition:undecidability:3-processes-invariants} does not harm the encoding of the 2CM.
	The overall idea is that, if a process ``chooses'' to not take some outgoing transition in the absence of invariants, then we show that the whole computation is stuck, and $\lochalt$ is unreachable.
	Put it differently, while the absence of invariants may add some runs, none of these runs can reach~$\lochalt$, and the final correctness argument of the proof of \cref{proposition:undecidability:3-processes-noinvariants} still holds.

	In the initial gadget in \cref{figure:2CM:initial:full}, if the process encoding~$\clockCMcone$ (resp.~$\clockCMctwo$) fails in taking the transition from $\locij{0}{\clockCMt}$ to~$\locij{1}{1}$ (resp.~from $\locij{0}{\clockCMt}$ to~$\locij{1}{2}$) after exactly $\param$ time units, then the first process encoding~$\clockCMt$ is stuck forever in~$\locij{1}{\clockCMt}$, and $\lochalt$ is unreachable.

	Now, consider the subpart encoding the clock~$\clockCMt$ (in all gadgets): in the absence of invariants, no transition is forced to be taken.
	However, since all outgoing transitions from the $\cmspta_i^{\clockCMt}$ locations are guarded by $\clockCMt = 0$, staying more than 0 time unit blocks this process forever, and $\lochalt$ becomes unreachable.
	The same applies to intermediate locations, which all have tight guards (\eg{} $\clockCMt = \param$).
	Therefore, to reach~$\lochalt$, then the process encoding clock~$\clockCMt$ must take all its guards at the appropriate time---which we assume from now on.

	Then, consider the two subparts encoding clocks $\clockCMcone$ and~$\clockCMctwo$.
	For the same reason, failing in taking a transition will block the process forever and, due to the final gadget in \cref{figure:2CM:final} ensuring transitions are guarded by $\clockx = 0$, will block the first process, and therefore $\lochalt$ will be unreachable.

	This concludes the proof.
\end{proof}
\subsection{Proof of \cref{proposition:2CM:fixedn}\label{appendix:proof:proposition:2CM:fixedn}}
\propositionTwoCMUndecFixed*
\begin{proof}
	The idea is to reuse the construction of the proof of \cref{proposition:undecidability:3-processes-noinvariants}, by ``killing'' any process except~3 by sending them to ``sink'' locations, from which they cannot interfere with the encoding of the 2-counter machine.
	Fix $n > 3$.
	The modified initial gadget is given in \cref{figure:2CM:initial:n}.
	The first three processes behave as in \cref{figure:2CM:initial:full}.
	Then, all remaining processes are sent to ``sink'' locations ($\locij{s}{3}$ to~$\locij{s}{n-1}$); for one process~$i$ to reach its sink location~$\locij{s}{i}$, its predecessor $i-1$ must have reached its own sink location~$\locij{s}{i-1}$, due to the location guard ``$\locguardof{\locij{s}{i-1}}$''.
	Finally note that the first process (simulating~$\clockCMt$) must make sure that the last process has reached its sink location (location guard ``$\locguardof{\locij{s}{n-1}}$'') before reaching the starting location~$\cmspta_0^{\clockCMt}$.
	This ensures that all processes $i > 3$ are in sink locations, where they cannot interfere with the encoding.

	The rest of the proof then follows the reasoning of the proof of \cref{proposition:undecidability:3-processes-noinvariants}.
\end{proof}
\begin{figure*}
{
	\centering
	\scalebox{\generalFigsScaleFactor}{
		\begin{tikzpicture}[PTA, node distance=2.5cm]
			\node[location, initial] at (0,0) (l0) {$\locij{0}{\clockCMt}$};
			\node[location, right of=l0] (l1) {$\locij{1}{\clockCMt}$};
			\node[location2CM, right=of l1] (s0) {$\cmspta_0^{\clockCMt}$};

			\node[location, below=of l1, yshift=5em] (l1_1) {$\locij{1}{1}$};
			\node[location2CM, right=of l1_1] (s0_1) {$\cmspta_0^1$};

			\node[location, below=of l1_1, yshift=5em] (l1_2) {$\locij{1}{2}$};
			\node[location2CM, right=of l1_2] (s0_2) {$\cmspta_0^2$};

			\node[location, sink, below=of l1_2, yshift=6em] (l1_3) {$\locij{s}{3}$};

			\node[below=of l1_3, yshift=6em] (l1_i) {$\cdots$};

			\node[location, sink, below=of l1_i, yshift=6em] (l1_n-1) {$\locij{s}{{n-1}}$};

			\path (l0) edge node[above, align=center]{$\clockx = \paramp \land \clockx > 1$} node[below]{$\clockx \assign 0$} (l1);
			\path (l1) edge node[above, locguard]{$\locguardof{\locij{s}{n-1}}$} node[below]{$\clockx = 0$} (s0);

			\path (l0) edge[bend right] node[above, locguard]{$\locguardof{\locij{1}{\clockCMt}}$} node[below,xshift=2em,yshift=-.5em]{$\clockx \assign 0$} (l1_1);
			\path (l1_1) edge node[above, locguard]{$\locguardof{\cmspta_0^{\clockCMt}}$} node[below, align=center]{$\clockx = 0$} (s0_1);

			\path (l0) edge[bend right] node[above, locguard]{$\locguardof{\locij{1}{1}}$} node[below,xshift=1em]{$\clockx \assign 0$} (l1_2);
			\path (l1_2) edge node[locguard]{$\locguardof{\cmspta_0^{\clockCMt}}$} node[below, align=center]{$\clockx = 0$} (s0_2);

			\path (l0) edge[bend right] node[below, locguard]{$\locguardof{\locij{1}{2}}$} (l1_3);

			\path (l0) edge[bend right] node[left, locguard]{$\locguardof{\locij{s}{{n-2}}}$} (l1_n-1);
		\end{tikzpicture}
	}

}
	\caption{Modified initial gadget for exactly $n$ processes}
	\label{figure:2CM:initial:n}
\end{figure*}
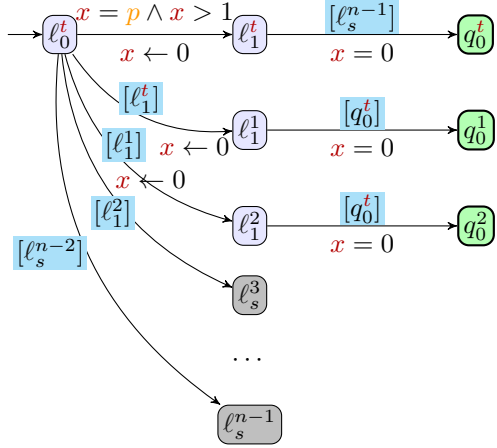
\subsection{Proof of \cref{theorem:undecidability-local-invariants}\label{appendix:proof:theorem:undecidability-local-invariants}}
\theoremUndecInv*
\begin{proof}
	First note that, in order to traverse the gadgets from the proof of \cref{proposition:undecidability:3-processes-invariants}, one needs (at least) 3~processes, due to the interdependency between the location guards.
	We will show that any process beyond the 3 required processes can only ``mimic'' the behaviour of one of the 3 required processes.

	Consider a process beyond the first three processes.
	In the initial gadget (\cref{figure:2CM:initial:full}), it can ``choose'' any of the three branches leading to the three subparts, as the first three processes make them all available.

	First assume this extra process follows the first subpart, simulating clock~$\clockCMt$: observe that, in all the gadgets of the first subpart (\cref{figure:2CM:initial:t,figure:2CM:increment:t,figure:2CM:decrement:t}), the behaviour is almost completely deterministic (\eg{} all locations with invariant ``$\clockCMt \leq \param$'' are followed by a transition guarded by ``$\clockCMt = \param$'').
	The only exception is in the branching to either $\locij{i}{\clockCMt}$ or~$\locij{k}{\clockCMt}$ in \cref{figure:2CM:decrement:t};
	assuming exactly one of the location guards $\locguardof{\locij{i1}{1}}$ or~$\locguardof{\locij{k}{1}}$ holds (which we discuss below), then this aforementioned extra process will follow exactly the behaviour of the first process.

	Alternatively, assume that this extra process follows the second subpart (simulating~$\clockCMcone$) and let us show that it will follow the second process %
		(the third subpart is similar).
	In most locations of our gadgets (\cref{figure:2CM:initial:x1,figure:2CM:increment:x1,figure:2CM:decrement:x1}), the outgoing guards are \emph{tight} \wrt{} the invariant (\eg{} invariant ``$\clockCMcone \leq \param + 1$'' followed by a transition guarded by ``$\clockCMcone = \param + 1$''), and this extra process is forced to follow the second process.
	In the remaining locations followed by non-tight guards (\eg{} ``$\clockCMcone \leq \param$'' in the increment gadget in \cref{figure:2CM:increment:x1}), the guard is always additionally guarded by the location guard $\locguardof{\cmspta_j^{\clockCMt}}$, which makes its behaviour deterministic since the processes simulating clock~$\clockCMt$ are deterministic.
	Therefore, this extra process is again forced to follow the second process.
	The 0-test and decrement gadget is similar, with one additional crucial observation: since the aforementioned extra process mimicked the second process so far, then in~$\cmspta_i^{1}$, due to the location guard~$\locguardof{\cmspta_i^{\clockCMt}}$ ensuring a transition in 0-time, then either $\clockCMcone = 0$ \emph{or} $\clockCMcone > 0$ holds, and therefore either both processes simulating~$\clockCMcone$ reach~$\locij{i1}{1}$, \emph{or} both of them reach~$\locij{k}{1}$---which justifies the assumption made earlier.

	For these reasons, $\lochalt$ is reachable only if the first three processes correctly simulate the 2CM;
	since any additional processes will just simulate one of the first three processes, we can just apply the correctness argument from the proof of \cref{proposition:undecidability:3-processes-invariants}.
\end{proof}
\subsection{Proof of \cref{theorem:undecidability-global-noinvariants}\label{appendix:proof:theorem:undecidability-global-noinvariants}}
\theoremUndecGlobal*
\begin{proof}
	We prove the result \emph{without} invariants.
	(The case with invariants is simpler and follows immediately.)

	The idea is that we ask whether \emph{all} processes can reach the target location~$\lochalt$, which can only be done if the 2CM was properly simulated.
	First note that asking whether \emph{all} processes can reach the target location~$\lochalt$ is a property satisfying the definition of global property in \cref{paragraph:definition:global}, by simply checking that the number of all locations but~$\lochalt$ is exactly~0 (called ``target problem'').

	As in the proof of \cref{theorem:undecidability-local-invariants}, in order to traverse the gadgets from the proof of \cref{proposition:undecidability:3-processes-invariants}, one needs (at least) 3~processes, due to the interdependency between the location guards.
	We show here that any process beyond the 3 required processes can only ``mimic'' the behaviour of one of the 3 required processes, or block the system due to a timelock if it fails to correctly mimic it.

	First note that if any of the first three processes does not correctly encode the 2CM (which is a possibility due to the absence of invariants), then an outgoing guard will not be firable, and therefore this process will be blocked forever in its current location, and the global property that all processes must reach~$\lochalt$ cannot hold.
	So we assume that the first three processes correctly encode the~2CM (\WLOG{} we assume that process~$i$ encodes the $i$-th subpart).

	Now consider a process beyond the first three processes.
	In the initial gadget (\cref{figure:2CM:initial:full}), it can ``choose'' any of the three branches leading to the three subparts, as the first three processes make them all available.
	Due to the absence of invariants, the process can also choose to not pass the initial gadget---but failing to take such a guard will block it forever in a location due to the guards $\clockx = 0$, and the global property that all processes must reach~$\lochalt$ will never hold.

	First assume this extra process follows the first subpart, simulating clock~$\clockCMt$:
	just as in the proof of \cref{theorem:undecidability-local-invariants},
	because in all the gadgets of the first subpart (\cref{figure:2CM:initial:t,figure:2CM:increment:t,figure:2CM:decrement:t}), the guards are always punctual (\ie{} involve equalities), then this extra process will either follow exactly the behaviour of the first process, or fail in taking some guard---therefore remaining forever in its location and violating the global reachability property.

	Alternatively, assume that this extra process follows the second subpart (simulating~$\clockCMcone$).
	As in the former reasoning, either that extra process will exactly mimic the behaviour of the second process, or will fail in taking some guard, and therefore being blocked in its location, thus violating again the global reachability property.

	The case of the final gadget is similar: if any of the processes arrive too early or too late, they will be blocked due to the urgent guards (of the form $\clockx = 0$ together with the location guards), and the global reachability will be violated.

	For these reasons, $\lochalt$ is reachable only if the first three processes correctly simulate the 2CM and if all additional processes simulate exactly one of the first three processes.
	Finally, we can again apply the correctness argument from the proof of \cref{proposition:undecidability:3-processes-noinvariants}.
\end{proof}
\section{Proofs of \cref{section:decidability}}
\subsection{Proof of \cref{theorem:decidability:fpPTA}\label{appendix:proof:theorem:decidability:fpPTA}}
\theoremDecidabilityfpPTA*
\begin{proof}
	Let $\PTA$ be a fully parametric \gPTA{} with 1~parameter~$\param$.
	$\Pnetwork{\PTA}{\infty}$ is therefore a fully parametric \PDTN{}.
	Let $\pval_0$ and~$\pval_1$ denote the valuations such that $\pval_0(\param) = 0$ and $\pval_1(\param) = 1$, a concept reused from the proof of \cite[Corollary~4.8]{HRSV02}.

	We prove the result for global reachability properties (PGR-emptiness), as local properties are a subcase.
	Fix a global property~$\varphi$.
	Let us show that PGR-emptiness does not hold iff $\varphi$ is satisfied in a configuration reachable in $\network{(\valuate{\PTA}{\pval_0})}{\infty}$ or in~$\network{(\valuate{\PTA}{\pval_1})}{\infty}$.

	The following lemma derives easily from \cite[Proposition~4.7]{HRSV02}, adapted to the semantics of \NTAs{} (\cref{semantics-NTAs}), and comes from the fact that, whenever no constant terms are used in a \gPTA{}, then rescaling the parameter valuation does not impact the satisfaction of reachability properties.

	\begin{lemma}[Multiplication of constants]\label{lemma:HRSV02:4.7}
		Let $\PTA$ be a fully parametric \gPTA{} with a single parameter~$\param$. %
		Fix $n \in \setN$.
		Let $\varphi$ be a global property.
		Then for all parameter valuations~$\pval$,
		a configuration $\nConfig$ with $\nConfig \models \varphi$ is reachable in~$\network{(\valuate{\PTA}{\pval})}{n}$
		iff
		$\forall t \in \setQpos$ such that $t \times \pval(\param) \in \setN$,
		a configuration $\nConfig'$ with $\nConfig' \models \varphi$ is reachable in~$\network{(\valuate{\PTA}{(t \times \pval)})}{n}$,
		where $t \times \pval$ denotes the valuation such that %
			$(t \times \pval)(\param) = t \times (\pval(\param))$.
	\end{lemma}

	We can now proceed to the proof of \cref{theorem:decidability:fpPTA}.

	\begin{description}
		\item[$\Rightarrow$]
			Assume PGR-emptiness does not hold for $\Pnetwork{\PTA}{\infty}$, \ie{} there exists $\pval$ such that there exists $n \in \setNpos$ such that $\varphi$ is satisfied in a reachable configuration in~$\network{(\valuate{\PTA}{\pval})}{n}$.
			Let us show that $\varphi$ is satisfiable in a configuration reachable in $\network{(\valuate{\PTA}{\pval_0})}{n}$ or in~$\network{(\valuate{\PTA}{\pval_1})}{n}$.

			If $\pval(\param) = 0$, then the result is immediate.
			If $\pval(\param) \neq 0$, then from \cref{lemma:HRSV02:4.7}, $\varphi$ is satisfied in a configuration reachable in $\network{(\valuate{\PTA}{\pval_1})}{n}$ (by choosing some appropriate $t$, \ie{} $\frac{1}{\pval(\param)}$).

		\item[$\Leftarrow$]
			Assume $\varphi$ is satisfied in a reachable configuration in~$\network{(\valuate{\PTA}{\pval_0})}{\infty}$ or in $\network{(\valuate{\PTA}{\pval_1})}{\infty}$.
			That is, there exists $n \in \setNpos$ such that
			there is a computation $\gcomputation$ of $\network{(\valuate{\PTA}{\pval_0})}{n}$ or of~$\network{(\valuate{\PTA}{\pval_1})}{n}$ reaching a configuration~$\nConfig$ s.t.\ $\nConfig \models \varphi$.
			Therefore PGR-emptiness does not hold.
	\end{description}
	Therefore, it suffices to test the satisfaction of $\varphi$ in $\network{(\valuate{\PTA}{\pval_0})}{\infty}$ and~$\network{(\valuate{\PTA}{\pval_1})}{\infty}$.

	Finally, the hardness argument is immediate, considering a \PDTN{} without parameter, and replacing constants different from~1 with additional clocks and locations.
\end{proof}
\subsection{Proof of \cref{theorem:L/U}\label{appendix:proof:theorem:L/U}}

Recall that $\ParamSet = \ParamSetL \uplus \ParamSetU$.
Given $\pval, \pval'$, we write $\pval' \decrincr \pval$ whenever
$\forall \param \in \ParamSetL, \pval'(\param) \leq \pval(\param)$
and
$\forall \param \in \ParamSetU, \pval'(\param) \geq \pval(\param)$.

\begin{lemma}[Monotonicity]\label{lemma:monotonicity-L/U-PDTN}
	Let $\PTA$ be an L/U-\gPTA{}.
	Let $\pval$ be a parameter valuation.
	For any $\pval'$ such that $\pval' \decrincr \pval$,
	for any~$n \in \setNpos$,
	any computation of $\network{(\valuate{\PTA}{\pval})}{n}$ is a computation of $\network{(\valuate{\PTA}{\pval'})}{n}$.
\end{lemma}
\begin{proof}
	From the fact that any valuation $\pval' \decrincr \pval$ will only \emph{add} behaviours due to the enlarged guards.
\end{proof}
\theoremDecidabilityLU*
\begin{proof}
	We prove the result for global reachability properties (PGR-emptiness), as local properties are a subcase.
	Let $\PTA$ be an L/U-\gPTA{} and $\varphi$ a global reachability property over~$\PTA$.
	$\Pnetwork{\PTA}{\infty}$ is therefore an L/U-\PDTN{}.
	Consider the \DTN{} $\network{(\valuate{\PTA}{\pvalzeroinf})}{\infty}$.
	Let us show that PGR-emptiness does not hold iff $\varphi$ is satisfied in a configuration reachable in $\network{(\valuate{\PTA}{\pvalzeroinf})}{\infty}$.

	\begin{description}
		\item[$\Rightarrow$]
			Assume PR-emptiness does not hold for $\Pnetwork{\PTA}{\infty}$, \ie{} there exists $\pval$ such that there exists $n \in \setNpos$ such that $\varphi$ is satisfied in~$\network{(\valuate{\PTA}{\pval})}{n}$.
			That is, there exists a computation~$\gcomputation$ of~$\network{(\valuate{\PTA}{\pval})}{n}$ reaching a configuration~$\nConfig$ such that $\nConfig \models \varphi$.
			From \cref{lemma:monotonicity-L/U-PDTN},
			$\gcomputation$ is a computation of $\network{(\valuate{\PTA}{\pval'})}{n}$, for any $\pval' \decrincr \pval$.
			And by extension, completely removing the upper-bound guards (\ie{} valuating upper-bound parameters with~$\infty$) only adds behaviour, and therefore $\gcomputation$ is a computation of $\network{(\valuate{\PTA}{\pvalzeroinf})}{n}$.
			Hence $\nConfig$ is reachable in $\network{(\valuate{\PTA}{\pvalzeroinf})}{\infty}$, and hence $\varphi$ is satisfied.

		\item[$\Leftarrow$]
			Assume there exists a configuration $\nConfig$ reachable in $\network{(\valuate{\PTA}{\pvalzeroinf})}{\infty}$ such that $\nConfig \models \varphi$.
		That is, there exists $n \in \setNpos$ such that there is a computation $\gcomputation$ of $\network{(\valuate{\PTA}{\pvalzeroinf})}{n}$ reaching a configuration~$\nConfig$ s.t.\ $\nConfig \models \varphi$.
		Now, $\pvalzeroinf$ is not a proper parameter valuation, so we need to exhibit a parameter valuation assigning to each parameter
		an integer value.
		We reuse the same concrete parameter valuation for upper-bound parameters as exhibited in \cite[Proposition~4.4]{HRSV02}: let $T'$ be the smallest constant occurring in the L/U-\gPTA{}~$\PTA$, and let~$T$ be the maximum clock valuation along~$\gcomputation$.
		Fix $D = T + |T'| + 1$.
		(We add $T'$ to compensate for potentially negative constant terms ``$d$'' in guards and invariants of~$\PTA$.)
		Since the maximum clock valuation along~$\gcomputation$ is~$T$, any guard of the form $\clock \leq \sum_{1 \leq i \leq \ParamCard} \alpha_i \times \parami{i} + d$, that was replaced with $\clock < \infty$ in $\network{(\valuate{\PTA}{\pvalzeroinf})}{n}$, can be equivalently replaced with $\clock \leq \sum_{1 \leq i \leq \ParamCard} \alpha_i \times D + d$ without harming the satisfaction of the guard.
		Therefore, $\nConfig$ is reachable in $\network{(\valuate{\PTA}{\pvalzerod{D}})}{n}$, and hence in $\network{(\valuate{\PTA}{\pvalzerod{D}})}{\infty}$.
		Therefore, since $\nConfig \models \varphi$, PGR-emptiness does not hold.
	\end{description}
	Therefore, deciding PGR-emptiness for L/U-\PDTNs{} amounts to deciding satisfaction of~$\varphi$ in the DTN $\network{(\valuate{\PTA}{\pvalzerod{D}})}{\infty}$.

	The case of local properties follows a similar reasoning.
\end{proof}
\subsection{Proof of \cref{theorem:decidability:no-invariant}\label{appendix:proof:theorem:decidability:no-invariant}}
\theoremDecidabilityNoInvariant*
In~\cite{AEJK24}, the shortest time to reach a location could be computed, allowing to replace the location guards one by one.
That is, for each location appearing in a location guard, we send one process to this location as quickly as possible, and which then remains in this location forever.
Hence, from the time this location can be reached, the location guard remains satisfied forever.
Note that this method only works thanks to the absence of invariants---which allows processes to ``die'' in every location.

However, this method cannot be reused here in the presence of parameters, as the notion of a ``shortest time'' is not entirely well-defined in this setting.
As such, we do not want to remove the location guards once at a time.
Instead, we will write a formula of the first order theory of the integers with addition (a.k.a.\ Presburger arithmetic) enhanced with the divisibility operand, which will include the computation of the reachable durations of the locations used in guards in parallel to the reachability of the final location.
Hence, solving PR-emptiness will boil down to deciding the truth of the formula.

In order to write this formula, we rely on results relating affine parametric semi-linear sets (apSl sets), a parametric extension of semi-linear sets, to durations in a PTA.
An apSl sets is a function associating to a vector of parameter values $\mathbf{p}$ a semi-linear set of vectors of integers.
We will not need the specific shape of apSl sets here, but instead two important properties they have.

First, given a PTA~$\PTA$ with one clock and arbitrarily many parameters %
and given two locations $\loc, \loc'$ of~$\PTA$,
one can compute~\cite{AAL24} the set of parametric durations of runs reaching $\loc'$
from~$\loc$, and represent it using a one-dimensional apSl set.
It is interesting to note that the apSl set representation only contain integers, while
the actual durations are a set of real values. The idea of the apSl representation is that
the value $2i$ is in $S(\mathbf{p})$ iff the integer $i$ is a reachable duration, and
$2i +1$ is in $S(\mathbf{p})$ iff all the values of the interval $(i,i+1)$ are reachable
durations. This representation hence strongly requires that the parameters range over
integers.
We also note that the construction of this apSl set can easily be modified so that
given two locations $\loc, \loc'$ and $n$ edges $t_1, \dots, t_{n}$ of~$\PTA$,
one could include the information of when the edge $t_1$ was crossed on the way to $\loc'$ for the first time.
For instance, if $n=1$, $(2i+1,2j)$ belongs to the corresponding apSl set iff $\loc'$ can be reached from $\loc$ in $j$ time units, and a path achieving this takes $t_1$
for the first time during the interval $(i,i+1)$.

Second, it was also shown in~\cite{AAL24} that given an apSl set~$S$, one can build a formula in the existential fragment of Presburger arithmetic with divisibility (a decidable logic~\cite{LOW15}) $\phi_S$ such that given parameter values $\mathbf{p}$, $S(\mathbf{p})$ is not empty iff $\exists x\in \phi_S(x,\mathbf{p})$ is true.

We can now move to the proof.

\begin{proof}
Let $\PTA = (\ActionSet, \LocSet, \locinit,
	\{x\}, \ParamSet, \invariant, \gEdgeSet)$ be a \gPTA{} and $\locfinal$ be a target location.

We first guess a sequence $e_1,\dots, e_m$ of different edges of~$\EdgeSet$ which contains a location guard (we trivially have that $m\leq |\EdgeSet|$).
These are intuitively the edges with location guards which will be needed, either
by the process reaching~$\locfinal$, or by the processes which will reach the locations used in location guards.
We assume these edges are ordered by the date at which they will be taken for the first time.
In practice, this guess can be achieved by complete enumeration.

For all $i \leq m,$ we set $\loc_i$ to be the location appearing in the guard of edge~$e_i$ and set $\loc_{m+1}=\locfinal$.
In order to reach~$\loc_i$, some edges from $e_1,\dots e_{i-1}$ may be needed.
Let $e_{j^r_1}, \dots, e_{j^r_{m_i}}$ be those edges, in the order they first appear
in the run.
We build the PTA $\PTA_i= (\ActionSet, \LocSet_i, (\locinit,0), (\loc_i, m_i), \{x\}, \ParamSet, \invariant_i, \EdgeSet_i)$ where
\begin{itemize}
\item $\LocSet_i = \big\{(\loc,k)\mid \loc \in \LocSet \wedge k \in \{0,\dots, m_i\} \big\}$,
\item for any $(\loc,k)\in \LocSet'$, $\invariant_i(\loc,k) = \invariant(\loc)$,
\item $\big((\loc,k), \guard, \action, \resets, (\loc',k') \big)\in \EdgeSet'$ iff
there exists $e=(\loc, \guard, \locguard, \action, \resets, \loc')\in \EdgeSet$ such that
one of the following holds
\begin{itemize}
\item $k=k'$ and $\locguard=\top$ or there exists $r \leq k$ such that $e=e_{j_r^i}$,
\item $k'=k+1$ and $e=e_{j^r_{k+1}}$.
\end{itemize}
\end{itemize}
In other words, $\PTA_i$ consists in $m_i+1$ successive copies of~$\PTA$
where the $k$'th~copy blocks every edge with location guard except
$e^i_{r}$ with $r\leq k+1$, and in particular taking $e^i_{k+1}$ leads to the next copy.
This way, a run of~$\PTA_i$ ending in the final location is forced to follow the guessed
structure with respect to the usage of edges with location guards.
It however loses information about when those edges are available.

By applying the previously mentioned result from~\cite{AAL24,Lefaucheux24},
we can build a semilinear set~$T_i$ of $m_i+1$-tuples parametric values
representing the durations of runs going from
$(\locinit,0)$ to~$(\loc_i, m_i)$,
storing the intervals of the first firing of the edges~$e_{j^{r_s}}$.

As previously mentioned, the construction of~$T_i$ does not take into account
the constraints brought by the location guards. Assuming that for $k<i$ location $\loc_k$
is reached at time~$h_i(\mathbf{p})$, we can see that deciding the existence of parameter values $\mathbf{p}$ such that $\loc_i$ is reached in $\PTA$ is equivalent to solving the
formula\footnote{The characters in bold are vectors of variables. In order to avoid
complexifying the formula, we did not indicate the transformation from row to column
vector which is necessary to multiply the two vectors and produce a single term.}
\begin{align*}
\exists \mathbf{p}, \exists \mathbf{d_{1}}, d'_1,\dots \mathbf{d_{m_i+1}}, d'_{m_i+1},
&(\mathbf{d_1}\mathbf{p}+d'_1,\dots, \mathbf{d_{m_i+1}}\mathbf{p}+d'_{m_i+1}) \in T_i, \\
&\bigwedge_{k=1}^{m_i} \mathbf{d_{j^i_k}}\mathbf{p}+d'_{j^i_k} \geq h_{j^i_k}(\mathbf{p})
\end{align*}
Moreover, note that the value $\mathbf{d_{m_i+1}}\mathbf{p}+d'_{m_i+1}$ built here is a possible value for~$h_{i+1}(\mathbf{p})$.

Hence by combining the formulas obtained for each~$i$, removing the comparison to~$h_i$ (which becomes redundant once the variables are shared by every formula), and verifying that the orders of each formula is compatible\footnote{For example, if the path to location $\loc$ goes through a location guard on~$\loc'$, then $\loc$ cannot reciprocally be on a location guard encountered on the way to~$\loc'$.
This condition on order is not directly handled in the formula in the case where both location guards can be reached within the same interval of time.
A~more precise decomposition of time units, would allow including this condition into the formula.}
we have that PR-emptiness is equivalent to the falsity of
\[
\exists \mathbf{p}, \exists \mathbf{d_{1}}, d'_1,\dots \mathbf{d_{m+1}}, d'_{m+1},
\bigwedge_{i=1}^{m+1}
(\mathbf{d_{j^r_1}}\mathbf{p}+d'_{j^r_1},\dots,
\mathbf{d_{j^r_{m_i}}}\mathbf{p}+d'_{j^r_{m_i}}, \mathbf{d_{i}}\mathbf{p}+d'_{i}) \in T_i.
\]
From~\cite{LOW15}, as this formula is expressed in the existential fragment of Presburger arithmetic with divisibility, it is decidable.
\end{proof}

\begin{remark}
Let us quickly discuss the complexity of this algorithm.
The formulas produced by~\cite{AAL24} are at worst doubly exponential.
The modifications we apply to them, combining a polynomial number of those formulas, remains
doubly exponential. We then rely on the decidability of
the existential fragment of Presburger arithmetic with divisibility
which can be solved in \NEXPTIME{}~\cite{LOW15}. The nondeterminism allowed through this last step combines with the nondeterministic guesses of transition sequences without additional cost.
As a consequence, our algorithm lies in \threeNEXPTIME{}.
\end{remark}

\end{document}